\documentclass[11pt]{article}
\usepackage{amsmath,amssymb,graphics}
\usepackage{bbold}
\usepackage{amsthm}
\usepackage{graphicx}
\usepackage{wasysym}
\usepackage[margin=2.5cm]{geometry}
\usepackage[all]{xy}
\usepackage{stmaryrd}
\usepackage{tabularx}
\usepackage{longtable}
\usepackage{xcolor}
\usepackage{todonotes}
\usepackage{tikz,tikz-3dplot}
\usetikzlibrary{calc,arrows,positioning}
\usepackage[T1]{fontenc}
\usepackage{url}
\usepackage[style=numeric-comp,sorting=anyt,defernumbers=true,maxnames=20,firstinits=true,uniquename=init,url=false,doi=false,isbn=false,eprint=false,backend=bibtex]{biblatex}
\renewbibmacro{in:}{%
  \ifentrytype{article}{}{\printtext{\bibstring{in}\intitlepunct}}}

\DeclareFieldFormat{labelnumberwidth}{\mkbibbold{#1.}}
\DeclareNameAlias{default}{last-first}

\usepackage{authblk}
\usepackage{lastpage}
\usepackage{fancyhdr}
\usepackage[linktocpage=false,linkcolor=black,colorlinks,pdfborder={0 0 0},citecolor=blue,pdfstartview={FitH}]{hyperref}

\newcommand\warning[1]{\textcolor{black}{#1}}

\theoremstyle{plain}
\newtheorem{thm}{Theorem}[section]
\newtheorem{lem}[thm]{Lemma}

\newtheorem{cor}{Corollary}[section]

\theoremstyle{definition}

\theoremstyle{remark}
\newtheorem{rem}{Remark}[section]

\numberwithin{equation}{section}
\numberwithin{figure}{section}

\allowdisplaybreaks

\title{Cohomological invariants of Abelian symplectic quotients of pure $r$-qubits}

\author[1]{Saeid Molladavoudi}

\date{\vspace{-5ex}}

\begin{document}

\maketitle

\begin{abstract} In this paper, we study cohomology rings and cohomological pairings over Abelian symplectic quotients of special Hamiltonian tori manifolds. The Hamiltonian group actions appear in quantum information theory where the tori are maximal tori of some compact semi-simple Lie groups, so called the Local Unitary groups, which act effectively and with non-trivial characters on specific complex projective varieties. The complex projective \warning{spaces} are in fact spaces of pure multi-partite quantum states. By studying the geometry of associated moment polytopes, we explicitly obtain cohomology rings, in terms of elementary symmetric functions, and a recursive wall-crossing algorithm to compute cohomological pairings over the corresponding Abelian symplectic quotients. We propose algorithms for general $r$-qubit case and elaborate discussions with explicit examples for the cases with $r=2,3$.
\end{abstract}

Keywords: Equivariant Cohomology; Symplectic Reduction; Multi-party Quantum States 


\section{Introduction} \label{introduction}

The study of multi-particle quantum correlations (entanglement) in quantum information theory is directly related to the moment map formalism in symplectic geometry (see \cite{benvegnu2004,christandl2014,saeid2015a,sawicki2011a} for example). Recall that the symmetry group of a composite quantum system containing $r$ isolated finite-dimensional sub-systems (spin systems for instance) in its pure state is denoted by the Local Unitary (LU) group, $K = SU(2) \times \cdots \times SU(2) \equiv SU(2)^r$, which is a semi-simple, compact Lie group acting in a Hamiltonian fashion on the complex projective space $M=\mathbb{P}_n$ of pure global states, such that $2^r = n+1$ \cite{ashtekar1995,brody2001}. The Lie group $K$ is a subgroup of Special Unitary group $SU(n+1)$. Therefore, there exists an equivariant moment map $\mu: \mathbb{P}_n \rightarrow \mathfrak{k}^*$, where $\mathfrak{k}^* \cong \oplus_{r} \mathfrak{su}^*(2)$ denotes the dual of the Lie algebra $\mathfrak{k}$ of the Lie group $K$, such that \cite{saeid2012b,sawicki2011b}
\begin{equation} \label{non-abelian-moment-map-eq}
\mu(p) = \bigoplus_{j=1}^{r}(\rho_{j}-\frac{1}{n_j}\mathbb{1}_{n_j}), \quad \forall p \in \mathbb{P}_n,
\end{equation}
where $\rho_j \in \mathfrak{su}(2)^*$ is the density matrix of $j$-th sub-system, or the $j$-th quantum marginal.

Similarly, the maximal torus $T^r$ of the LU group $K$, which is denoted by $T^r=(S^1)^r$ acts on the K\"ahler manifold $M=\mathbb{P}_n$ in a Hamiltonian way and so $(\mathbb{P}_n,\omega_{FS},T^r,\mu_T)$ is also a Hamiltonian $T^r$-manifold, where $\omega_{FS}$ denotes the Fubini-Study symplectic form on $\mathbb{P}_n$ and $\mu_T:\mathbb{P}_n \rightarrow \mathfrak{t}^*_r$ denotes an equivariant moment map whose moment values $\xi=\mu_T(p)$, for $p \in \mathbb{P}_n$, are collections of diagonal elements of the quantum marginals. \warning{This Hamiltonian $T^r$ action on $\mathbb{P}_n$, such that $n = 2^r -1$ (please see Eq. \eqref{action} below), is the main focus of the current paper.}

The process of symmetry reduction of co-adjoint orbits as Hamiltonian manifolds, also known as symplectic reduction, has been studied numerously in the past decades (\cite{dh1982,marsden1974,guillemin1982,kirwan1984,ginzburg2002} to mention a very few). An underlying theme has been the problem of understanding the geometrical and topological properties of the resulting symplectic reduced spaces (also known as symplectic quotients). More precisely, in the context of symmetry reduction of a Hamiltonian system $(M,\omega,K,\mu)$, the components of the equivariant moment map $\mu: M \rightarrow \mathfrak{k}^*$ are conserved with respect to integral curves of Hamiltonian vector fields, i.e. the \textit{Noether's theorem}. If the isotropy subgroup $K_{\xi} = \left\{ g \in K | \, \text{Ad}^*_g(\xi) = \xi, \, \xi \in \mu(M) \subset \mathfrak{k}^* \right\}$, where $\text{Ad}^*_g$ denotes the coadjoint action of the Lie group on $\mathfrak{k}^*$, acts freely on $\mu^{-1}(\xi)$ then $\mu^{-1}(\xi)$ is a submanifold of $M$ and the quotient $\mu^{-1}(\xi)/K_{\xi}$ possesses an induced symplectic form and is called the Marsden-Weinstein quotient \cite{marsden1974}. However, the symplectic quotient $\mu^{-1}(\xi)/K_{\xi}$ acquire singularities if $\xi$ is not a regular value of $\mu$ or if $K_{\xi}$ does not act freely on $\mu^{-1}(\xi)$, which is called the singular reduction of Hamiltonian manifolds \cite{sjamaar1991}.

Moreover, for every Hamiltonian $K$-manifold $(M,\omega,K,\mu)$, which is equipped with an equivariant moment map $\mu:M \rightarrow \mathfrak{k}^*$, there exists a \textit{convex} polytope $\Delta := \mu^{\prime}(M)$, so called the moment (Kirwan) polytope \cite{guillemin1982,kirwan1984}. The corresponding \textit{invariant moment map} $\mu^{\prime}$ is defined by $\mu^{\prime}: M \rightarrow \mathfrak{t}^*_+, p \mapsto \mu^{\prime}(p) = \mu(K\, . \, p) \cap \mathfrak{t}^*_+$, in which $\mathfrak{t}^*_+ = \mathfrak{k}^*/K$ denotes the positive Weyl chamber. Symplectic reduction of Hamiltonian $K$-manifolds has had major applications thus far in various quantization techniques, such as Berezin-Toeplitz deformation quantization \cite{schlich1994}, and in some proofs of Guillemin-Sternberg's conjecture ``quantization commutes with reduction'', where the symplectic manifold $(M,\omega)$ denotes the phase space of a classical system acted upon symplectically by a semi-simple, compact Lie group $K$ \cite{schlich2001a}. The symplectic quotients $M_{\xi}= \mu_T^{-1}(\xi)/T^r$ of a Hamiltonian torus manifold $(M,\omega,T^r,\mu_T)$, for a regular value $\xi \in \Delta_{\text{reg}}$ where $\mu_T(\mathbb{P}_n) \equiv \Delta$, possess at most finite-quotient (or the so-called orbifold) singularities \cite{kirwan1984,kirwan1985}. Recall that $\mu_T(M)$ is the convex hull of the finite set of points $\mu_T(M^T)$, where $M^T$ denotes the fixed point set of the $T^r$-action on $M$ \cite{atiyah1982}.

\warning{In this paper}, the torus $T^r=(S^1)^r$ is the maximal torus of the LU group $K$ and its action on $\mathbb{P}_n$ is diagonalizable through the following homomorphism $\varphi: T^r \rightarrow U(n+1)$ as follows \cite{saeid2015a}
\begin{equation} \label{action}
\warning{\varphi: (t_1,t_2, \cdots , t_r) \mapsto \varphi (t_1,t_2, \cdots , t_r) = \left( \prod_{i=1}^{r}t_i^{a_{i,1}}, \prod_{i=1}^{r}t_i^{a_{i,2}}, \cdots, \prod_{i=1}^{r}t_i^{a_{i,n+1}}\right),}
\end{equation}
for some $a_{i,j} \in \left\{ \pm 1 \right\}$, where $\sum_{k=0}^{r}{\binom{r}{k}}=2^r=n+1$ and the $r\times(n+1)$ matrix $A$ contains all $a_{i,j}$s as columns.

The maximal torus $T^r$ of the Lie group $K$ is a subgroup of the maximal torus $T^{n+1}$ of the Lie group $SU(n+1)$, for which the components of the moment map $\Phi: M \rightarrow \mathfrak{t}^*_{n+1}$ are \textit{perfect Morse functions} with critical set $M^{T^{n+1}}$. Using Morse-theoretic arguments, one can easily show that
\[
M^{T^r} = \left\{ [0: \cdots:z_j: \cdots:0] \in \mathbb{P}_n | \, z_j=1, \, j=1, \cdots, n+1 \right\} = M^{T^{n+1}},
\]
where the Abelian  moment map $\mu_T: \mathbb{P}_n \rightarrow \mathfrak{t}^*_r$ for the $T^r$-action is given by
\begin{equation} \label{moment-map}
\mu_T(p) = \frac{1}{2} \sum_{j=1}^{n+1}{|z_j|^2 \alpha_j},
\end{equation}
and where $\alpha_j=(a_{1,j},a_{2,j}, \cdots,a_{r,j})^T$, for $j=1, \cdots,n+1$ are weights of the representation of $T^r$ on $\mathbb{C}^{n+1}$ \cite{kirwan1984}. The moment polytope $\mu_T(\mathbb{P}_n)$ is then a hypercube 
\begin{equation} \label{hyper-r-cube-eq}
\Delta \equiv \text{conv}.\left\{(x_1,\cdots,x_r) \in \mathbb{R}^r : \, x_i=\pm 1/2 \right\},
\end{equation}
spanned by $2^r$ vertices as the image of the fixed points under the moment map. The codimension-one walls divide the moment polytope $\Delta$ into subpolytopes $\Delta_{\text{reg}}$, whose interior consists of entirely regular values of the Abelian moment map $\mu_T$.

The \textit{equivariant cohomology} $H^*_K(M)$ provides us with some type of cohomology that is defined for (Hamiltonian) $K$-spaces in which not all points have trivial isotropy subgroups. There are two approaches to equivariant cohomology: topological approach (due to Borel \cite{borel1960}) and algebraic approach (due to Cartan in 1950). In Cartan's model, one can define an analogue of the de Rham approach to the \textit{equivariant differential forms} on $M$. Let $\mathfrak{k}$ be the Lie algebra of $K$ and $\Omega^*(M)$ denotes the algebra of smooth differential forms on $M$. Recall that a $K$-equivariant differential form is a map $\eta : \mathfrak{k} \rightarrow \Omega^*(M)$ such that
\[
\eta(X) := g^{-1} \, . \, \eta(\text{Ad}_g X), \quad \forall \, X \in \mathfrak{k} , \, g \in K.
\]
If $\mathbb{C}[\mathfrak{k}]$ denotes the algebra of complex polynomials on $\mathfrak{k}$, the algebra of equivariant differential forms $\Omega^*_K(M)$ is given by the subalgebra of $K$-invariant forms $(\mathbb{C}[\mathfrak{k}] \otimes \Omega^*(M))^K$. We can also define the $K$-equivariant exterior differential $D$ by $(D \eta)(X) := d(\eta(X)) - \iota_{X^{\#}}(\eta(X))$, for all $X \in \mathfrak{k}$, and where $\iota_{X^{\#}}$ denotes the contraction with vector field $X^{\#}$ induced by the infinitesimal action of $\mathfrak{k}$ on $M$. What makes equivariant forms important is that they form a complex $(\Omega^*_K(M),D)$, which in turn can be related to the ordinary cohomology. For compact, connected Lie group $K$ acting on smooth, compact manifold $M$, we have the equivariant de Rham cohomology $H^*_K(M) \equiv H^*(\{ \Omega^*_K(M),D \})$.

One of the features of equivariant cohomology is the \textit{localization formula} for the integrals of equivariantly closed differential forms. Recall the Atiyah-Bott, Berline-Vergne (ABBV) localization theorem for a Hamiltonian toric manifold, which relates the integral of an equivariantly closed form on manifold $M$ to the fixed-point set of the associated torus action. More precisely, when a torus $T=(S^1)^r$ acts on a compact, oriented manifold $M$ with $2r = \text{dim}_{\mathbb{R}}(M)$ and fixed point set $F$, for any closed $T$-equivariant form $\eta \in \Omega^*_T(M)$, such that $(D \eta)(\, . \, ) = 0$, we have
\begin{equation}
\int_M{\eta(X)} = \sum_{f \in F}{\int_f{\frac{\imath^*_f \, \eta(X)}{e_f(X)}}},
\label{abelian-localization-eq}
\end{equation}
where $\imath^*_f \, \eta(X)$ is the inclusion of the fixed component $f \in F$ and $e_f \in H^*_T(M)$ is the equivariant Euler class of the normal bundle to $f$ in $M$. \warning{In the literature}, Eq. \eqref{abelian-localization-eq} is called the  \textit{Abelian localiztion formula} \cite{atiyah1983,vergne1982,vergne1983,vergne1992}.

The same is true for any $D$-closed formal series $\sum_j \eta_j$, for $\eta_j \in \Omega^j_T(M)$, such as those of the form $\eta e^{i\tilde{\omega}}$, where $\tilde{\omega} :=\omega + \mu_X \in \Omega^2_T(M)$ is the equivariantly closed symplectic form on $M$ and $\mu_X$ is the Hamiltonian function generated by the corresponding equivariant moment map for every $X \in \mathfrak{t}$. By substituting the equivariant cohomology class $\eta e^{i\tilde{\omega}}$ in Eq. \eqref{abelian-localization-eq} and considering the fact that the functions $\mu_X(f) \in \mathfrak{t}^*$ are constant for fixed component $f$ we obtain the Duistermaat-Heckman theorem \cite{dh1982}, according to which the integral $\int_M{\exp{(-ti\mu_X)} \, \frac{\omega^n}{n!}}$ is exactly given by its stationary phase approximation. 

In \cite{jeffrey1995}, Jeffrey and Kirwan extended the localization for non-Abelian cases and developed a systematic approach to determine the cohomology of Marsden-Weinstein reduced space $\mu^{-1}(\xi)/K_{\xi}$ from the corresponding equivariant cohomology of the original manifold $M$ under compact group $K$ action. Considering the Kirwan surjective ring homomorphism $\kappa: H^*_K(M) \rightarrow H^*(\mu^{-1}(\xi)/K_{\xi}), \tilde{\omega} \mapsto \omega_{\xi}$ \cite{kirwan1984}, the resulting residue formula is explicitly obtained in \cite{jeffrey1995,jeffrey1997,jeffrey1998}. In \cite{jeffrey2003}, the residue formula for singular symplectic and their associated GIT quotients are obtained by partial desingularization of reduced spaces \cite{kirwan1985}. Non-Abelian localization has had two major applications thus far: first, for intersection numbers in moduli spaces of vector bundles on Riemannian surfaces \cite{jeffrey1998,kiem2000} and second, in some proofs of the Guillemin-Sternberg's conjecture that ``quantization commutes with reduction'' \cite{jeffrey1997,meinrenken1999}.

In this paper, we are utilizing a slightly different version of the localization, or the so-called iterated residue algorithm, introduced in \cite{guillemin1996} to compute integrals of the Kirwan's map of the equivariant cohomology classes $\tilde{\omega}^m \equiv \tilde{\omega} \wedge \cdots \wedge \tilde{\omega}$, where $\tilde{\omega} \in H^{2}_T(M,\mathbb{C})$, over the Abelian symplectic reduced spaces $M_{\xi} \cong \mu_T^{-1}(\xi)/T$, $\xi \in \Delta_{\text{reg}}$, of the Hamiltonian $T^r$-manifold $(\mathbb{P}_n,\omega_{FS},\mu_T,T^r)$, such that $2^r=n-1$.  From a straightforward dimensional analysis we find that $m=2^{r}-(r+1) = \text{dim}_{\mathbb{C}}(M_{\xi})$. Unless otherwise stated, throughout this paper $\omega$ denotes the Fubini-Study symplectic form $\omega_{FS}$ on the K\"{a}hler manifold $\mathbb{P}_n$.

The first step is to find the kernel of the Kirwan map $\kappa: H^*_T(M) \rightarrow H^*(\mu_T^{-1}(\xi)/T)$. Recall that the Kirwan map $\kappa: H^*_T(M) \rightarrow H^*(M_{\xi})$ is a surjective map. Therefore, finding kernel of the Kirwan map is equivalent to finding cohomology ring of the reduced space $M_{\xi}$, i.e. $H^*(M_{\xi})$, since the equivariant cohomology ring $H^*_{T^r}(\mathbb{P}_n;\mathbb{C})$ of a complex projective manifold $\mathbb{P}_n$ under the Hamiltonian action of $T^r$\warning{, as in Eq. \eqref{action},} is known. The second step is to use the combinatorial object (called ``dendrite'') and the wall-crossing algorithm introduced in \cite{guillemin1996} to iteratively express $\int_{M_\xi}\kappa(\tilde{\eta}^a\tilde{\omega}^b)$, where $\tilde{\eta}$ and $\tilde{\omega}$ are degree two complementary equivariant cohomology classes such that  $a+b=m=\text{dim}_{\mathbb{C}}(M_{\xi})$, in terms of integration over the connected components of $T^r$ fixed point set $M^{T} \cong F$.

The outline of paper is as follows. In section \ref{cohomology-rings} we generalize the results of \cite{kalkman1995}, from circle to a torus action, and find the cohomology rings $H^*(M_{\xi})$ of corresponding Abelian symplectic quotients $M_{\xi}$ in terms of elementary symmetric functions for the cases $r=2,3$ and propose an algorithm for a general $r$-case. In section \ref{recursive-wall-crossing-integrals}, by studying the geometry of the corresponding Kirwan polytope of the Hamiltonian $T^r$-manifold $(\mathbb{P}_n,\omega,\mu_T,T^r)$ we investigate the wall-crossing algorithm \warning{and use} dendrites to compute cohomological pairings over the Abelian symplectic reduced spaces and elaborate the discussions with explicit examples for the cases with $r=2,3$. Finally, in section \ref{conclusions} we conclude and provide outlooks for future studies.

\section{Cohomology Rings of Abelian Quotients} \label{cohomology-rings}
\subsection{Special Case: Circle Action} \label{circle-action}
In \cite{tolman2003}, Tolman and Weitsman proposed a method to characterize kernel of the Kirwan map $\kappa$
\begin{equation} \label{kirwan-map-eq}
\kappa: H_T(M) \rightarrow H^*(M_{\xi}), \quad \xi \in \mathfrak{t}^*,
\end{equation}
for a Hamiltonian action of a compact Lie group $T$ on a symplectic manifold $(M,\omega)$. Their method is based on the cohomology of the original symplectic manifold and the data encoded in fixed point set of the action. More precisely, for the case of a circle action:
\begin{thm}
\cite{tolman2003} Consider the Hamiltonian torus $T$-manifold $(M,\omega,T,\mu_T)$, where $(M,\omega)$ is a compact symplectic manifold. Assume $0$ is a regular value of moment map $\mu_T$ and $F$ denotes the fixed point set. For all $X \in \mathfrak{t}$ define
\[
N_X := \left\{p \in M | \, \langle \mu_T(p),X \rangle \leq 0 \right\},
\] 
\[
Q_{X} := \left\{\alpha \in H^*_T(M;\mathbb{Q}) | \, \alpha|_{F \cap N_X} =0 \right\},
\]
\[
Q = \sum_{X \in \mathfrak{t}}{Q_X}.
\]
Then, there is a short exact sequence
\[
0 \rightarrow Q \rightarrow H^*_T(M;\mathbb{Q}) \xrightarrow{\kappa} H^*(\mu_T^{-1}(0)/T;\mathbb{Q}) \rightarrow 0
\]
where $\kappa: H^*_T(M;\mathbb{Q}) \rightarrow H^*(\mu_T^{-1}(0)/T;\mathbb{Q})$ is the Kirwan map.
\end{thm}
It can also read as follows \cite{jeffrey2003b}:
\begin{cor} \label{tolman-weitsman-circle-cor}
For the case of a circle action $T=S^1$, the kernel $Q=Ker(\kappa_S)$ of the Kirwan map $\kappa_S$ is the sum
\begin{equation} \label{tw-kernel-circle}
Q \equiv Ker(\kappa_S) = Q_+ \oplus Q_{-},
\end{equation}
where $Q_{\pm} = \left\{ \alpha \in H^*_T(M) : \, \alpha|_f = 0 \, \, \text{for all} \, \, f \in F \, \text{for which} \, \pm \mu_T(f) > 0 \right\}$.
\end{cor}

In \cite{kalkman1995}, Kalkman obtained the cohomology ring of symplectic quotient of a projective \warning{space} under the \warning{Hamiltonian} action of a circle \warning{$S^1$}. He obtained the following isomorphism of algebras
\begin{equation} \label{kalkman-isomorphism-algebras}
H^*_{S^1}(\mathbb{P}_n) \cong \mathbb{C}[\phi,\omega]/ ( \prod_{p \in F} \left( \omega + \mu_T(p) \phi \right) ),
\end{equation}
for the equivariant cohomology of a Hamiltonian $S^1$-manifold $(\mathbb{P}_n, \omega,\mu_T,S^1)$, with fixed point set $F$ and where the cohomology groups are with complex coefficients. In Eq. \eqref{kalkman-isomorphism-algebras} $\omega$ is the degree-two symplectic form $\omega \in H^2(\mathbb{P}_n)$ generating the cohomology $H^*(\mathbb{P}_n)$ and $\phi$ is the degree-two generator of the equivariant cohomology $H^*_{S^1}(pt)$.
\begin{thm} \label{kalkman-thm}
\cite{kalkman1995} Let $I \subset \mathbb{C}[\phi,\omega]$ be the ideal generated by 
\begin{equation} \label{kalkman-map-eq}
\tilde{Q}_{+} = \prod_{ \mu_{S^1}(p) > \xi } \left( \omega + \mu_{S^1}(p) \phi \right),
\quad \tilde{Q}_{-} = \prod_{ \mu_{S^1}(p) < \xi } \left( \omega + \mu_{S^1}(p) \phi \right).
\end{equation} 
Then, for the Hamiltonian circle $S^1$-manifold $(\mathbb{P}_n, \omega,\mu_{S^1},S^1)$, with fixed point $p \in F$, we have the following isomorphism of rings
\begin{equation} \label{kalkman2}
H^*(\mu_{S^1}^{-1}(\xi)/S^1) \cong \mathbb{C}[\phi,\omega]/I.
\end{equation}
\end{thm}
In fact, Kalkman obtained the ring isomorphism \eqref{kalkman2} by composing the surjective algebra homomorphism $h_S: \mathbb{C}[\phi,\omega] \rightarrow H^*_{S^1}(\mathbb{P}_n)$ and the Kirwan map $\kappa_S: H^*_{S^1}(\mathbb{P}_n) \rightarrow H^*(\mu_{S^1}^{-1}(\xi)/S^1)$. In other words, $\tilde{Q}_{+}$ and $\tilde{Q}_{-}$ are generators of the kernel of surjective homomorphism $\kappa_S \circ h_S:\mathbb{C}[\phi,\omega] \rightarrow H^*(\mu_{S^1}^{-1}(\xi)/S^1)$ \cite{mohammadalikhani2004}.

\subsection{Torus $T^r$-Action}
In this section we denote a $1$-dimensional torus (circle) by $S$ and a $r$-dimensional torus by $T$. The following theorem asserts that the knowledge of kernel of the Kirwan map $\kappa_S: H^*_S(M) \rightarrow H^*(\mu_S^{-1}(\xi)/S)$, where $\xi \in \mathfrak{s^*}$ is a regular value of the moment map $\mu_S: M \rightarrow \mathfrak{s}^*$ for a generic circle $S \subset T$, can be used to determine the kernel of the Kirwan map $\kappa$ in Eq. \eqref{kirwan-map-eq}:
\begin{thm} \label{kernel-kirwan-map-jeffrey-thm}
\cite{tolman2003,jeffrey2005}
\begin{equation}
Ker(\kappa) = \sum_{S\subset T}{Q_+^S \oplus Q_{-}^S},
\end{equation}
where $Q_+^S \oplus Q_{-}^S = Ker(\kappa_S)$ is the Tolman-Weitsman kernel for circle action in Eq. \eqref{tw-kernel-circle} in corollary \ref{tolman-weitsman-circle-cor} and the sum is over all generic circles, such that two fixed point sets $M^S$ and $M^T$ are equal and $\xi \in \mathfrak{s}^*$ is a regular value of the moment map $\mu_S$. 
\end{thm}

\begin{rem}
In general, for an $r$-dimensional torus $T=(S)^r$ acting on a compact manifold $M$, the characters of the action can be identified with cohomology classes of  equivariant cohomology of $M$ \cite{zielenkiewicz2014}. More precisely, the kernel of the algebra $h: H^*_{T^r}(\text{pt})[\omega] \cong \mathbb{C}[\theta_1,\theta_2, \cdots,\theta_r,\omega] \rightarrow  H^*_T(\mathbb{P}_n)$ is the ideal generated by the relation 
\begin{equation} \label{relation-kernel-eq}
\prod_{p \in F}{(\omega + \chi_j)}=0,
\end{equation}
where the $\chi_j$s are distinct additive characters of the torus $T^r$-action on $\mathbb{P}_n$ \cite{weber2012}. Using the elementary symmetric functions, $\sigma_j$, the relation becomes
\begin{equation} \label{symmetric-pol-eq1}
\sum_{j=1}^{n+1}{\sigma_j({\chi_1, \cdots, \chi_{n+1}}) \,  \omega^{n-j+1}} = 0.
\end{equation}
\end{rem}
\begin{rem}
Recalling the Eq. \eqref{action}, in our case we have
\begin{equation} \label{equivariant-cohomo-pn}
H_{T^r}^*(\mathbb{P}_n) \cong \mathbb{C}[\theta_1,\theta_2, \cdots,\theta_r,\omega] / \langle \prod_{j=1}^{n+1}{ \left( \omega + \sum_{i=1}^{r}{a_{i,j}\theta_i} \right) } \rangle,
\end{equation}
where $\chi_j=\sum_{i=1}^{r}{a_{i,j}\theta_i}$, for $j=1,\cdots,n+1$ are additive characters of the torus $T^r$-action and $\omega$ is the degree-two generator of the cohomology ring $H^*(\mathbb{P}_n)$ and $H^*_{T^r}(\text{pt})[\omega] \cong \mathbb{C}[\theta_1,\theta_2, \cdots, \theta_r,\omega]$. 
\end{rem}

\begin{cor} \label{torus-kalkman-cor}
Let $\tilde{Q}_T \subset \mathbb{C}[\theta_1,\theta_2,\cdots,\theta_r,\omega]$ be the ideal generated by
\begin{equation} \label{torus-generators-eq1}
\tilde{Q}_T = \sum_{S \subset T}{\tilde{Q}^S_{+} \oplus \tilde{Q}^S_{-}},
\end{equation}
where
\begin{equation} \label{torus-generators-eq2}
\tilde{Q}^S_{+} = \prod_{\substack{p \in \mathcal{F} \\ \langle \mu_S(p), \alpha \rangle > \langle \xi,\alpha \rangle}}{ \left( \omega + \sum_{i=1}^{r}{a_{i,j}\theta_i} \right) }, \quad \alpha \in \mathfrak{s}^*, 
\end{equation}
\begin{equation} \label{torus-generators-eq3}
\tilde{Q}^S_{-} = \prod_{\substack{p \in \mathcal{F} \\ \langle \mu_S(p), \alpha \rangle < \langle \xi,\alpha \rangle}}{ \left( \omega + \sum_{i=1}^{r}{a_{i,j}\theta_i} \right) }, \quad \alpha \in \mathfrak{s}^*. 
\end{equation}
Then, $\tilde{Q}_T$ is the kernel of the surjective homomorphism $\kappa \circ h: \mathbb{C}[\theta_1,\theta_2,\cdots,\theta_r,\omega] \rightarrow H^*(M_{\xi})$.
\end{cor}
\begin{proof}
Considering theorems \ref{kalkman-thm} and \ref{kernel-kirwan-map-jeffrey-thm}, the only remaining part of the proof is to show that $Ker(h) = \sum_{S \subset T}{Ker(h_S)}$, where $h_S: H^*_S(pt)[\omega] \cong \mathbb{C}[\phi,\omega] \rightarrow H^*_{S^1}(\mathbb{P}_n)$. The latter is evident by writing the relation in Eq. \eqref{equivariant-cohomo-pn} in terms of elementary symmetric functions $\sigma_j$, as in Eq. \eqref{symmetric-pol-eq1}. In other words,
\begin{equation}
Ker(\kappa \circ h) = \sum_{S \subset T}{Ker(\kappa_S \circ h_S)} = \sum_{S \subset T}{\tilde{Q}^S_{+} \oplus \tilde{Q}^S_{-}}, 
\end{equation}
where $\tilde{Q}^S_{+}$ and $\tilde{Q}^S_{-}$ are given in Eqs. \eqref{torus-generators-eq2}, \eqref{torus-generators-eq3}, with characters $\chi_j=\sum_{i=1}^{r}{a_{i,j}\theta_i}$, for $j=1,\cdots,n+1$.

\end{proof}
\begin{rem}
For the action \eqref{action}, $p \in F$ actually means that for each $S \subset T$ we have exactly $2^{r-1}$ of the images of fixed points (vertices of the hypercube $\Delta$) falling into either side ($>$ or $<$) of the hyperplane parallel to the co-dimension one walls of the hypercube (i.e. the moment polytope) containing $\xi \in \Delta_{\text{reg}}$.
\end{rem}

\subsection{Examples}
Now, let's consider the case $r=2$, namely the torus $T^2=S \times S$ acting in a Hamiltonian fashion on the K\"ahler manifold $M=\mathbb{P}_3$ via the homomorphism \eqref{action}. We can find generators of the kernel of surjective homomorphism $\kappa \circ h: \mathbb{C}[\theta_1,\theta_2,\omega] \rightarrow H^*(M_{\xi})$ as follows:
\begin{itemize}
\item Recall that the matrix $a_{i,j}\equiv A$ is given by
\begin{equation} \label{a-matrix-r=2}
A \equiv a_{i,j} = \left( \begin{matrix}
1 & 1 & -1 & -1 \\
1 & -1 & 1 & -1
\end{matrix}\right)_{2 \times 4}.
\end{equation}
\item The moment polytope is a square (2-cube) with 4 vertices as the image under the moment map $\mu_T : \mathbb{P}_3 \rightarrow \mathfrak{t}^*$. The symplectic quotient $M_{\xi} = \mu_T^{-1}(\xi)/T^2$, with $\xi \in \Delta_{\text{reg}}$ possesses at most orbifold singularities. 
\item Recalling the corollary \ref{torus-kalkman-cor} and Eq. \eqref{a-matrix-r=2}, for each 1-dimensional torus (circle) $S$ we have
\begin{itemize}
\item For first circle $S$ action with $\alpha_1 \in \mathfrak{s}^*$:
\begin{equation} \label{kernel-generators-r=2-eq1}
\tilde{Q}^S_{+,1} = \prod_{\substack{p \in \mathcal{F} \\ \langle \mu_S(p), \alpha_1 \rangle > \langle \xi,\alpha_1 \rangle}}{ \left( \omega + \sum_{i=1}^{2}{a_{i,j}\theta_i} \right) }=(\omega + \theta_1 + \theta_2)(\omega + \theta_1 - \theta_2),
\end{equation}
\begin{equation} \label{kernel-generators-r=2-eq2}
\tilde{Q}^S_{-,1} = \prod_{\substack{p \in \mathcal{F} \\ \langle \mu_S(p), \alpha_1 \rangle < \langle \xi,\alpha_1 \rangle}}{ \left( \omega + \sum_{i=1}^{2}{a_{i,j}\theta_i} \right) }=(\omega - \theta_1 + \theta_2)(\omega - \theta_1 - \theta_2).
\end{equation}
\item For second circle $S$ action with $\alpha_2 \in \mathfrak{s}^*$:
\begin{equation} \label{kernel-generators-r=2-eq3}
\tilde{Q}^S_{+,2} = \prod_{\substack{p \in \mathcal{F} \\ \langle \mu_S(p), \alpha_2 \rangle > \langle \xi,\alpha_2 \rangle}}{ \left( \omega + \sum_{i=1}^{2}{a_{i,j}\theta_i} \right) }=(\omega + \theta_1 + \theta_2)(\omega - \theta_1 + \theta_2),
\end{equation}
\begin{equation} \label{kernel-generators-r=2-eq4}
\tilde{Q}^S_{-,2} = \prod_{\substack{p \in \mathcal{F} \\ \langle \mu_S(p), \alpha_2 \rangle < \langle \xi,\alpha_2 \rangle}}{ \left( \omega + \sum_{i=1}^{2}{a_{i,j}\theta_i} \right) }=(\omega + \theta_1 - \theta_2)(\omega - \theta_1 - \theta_2).
\end{equation}
\end{itemize}
\item In terms of elementary symmetric polynomials $\sigma_j$ in Eq. \eqref{symmetric-pol-eq1} and recalling the fact that $\chi_1=\theta_1 + \theta_2, \, \chi_2 = \theta_1 - \theta_2, \, \chi_3 = -\theta_1 + \theta_2, \, \chi_4 =- \theta_1 - \theta_2$, we have
\begin{eqnarray*}
\tilde{Q}^S_{+,1} & = & \sigma_1(\chi_1,\chi_2,\chi_3,\chi_4) \, \omega^{3} + \sigma_2(\chi_1,\chi_2,\chi_3,\chi_4) \, \omega^{2} \\ \nonumber
& = & \omega^2 \left(\omega  \chi_3+\omega  \chi_4+\chi_2 \left(\omega +\chi_3+\chi_4\right)+\chi_1 \left(\omega +\chi_2+\chi_3+\chi_4\right)+\chi_4 \chi_3\right) \\ \nonumber
& = & -2 \, \omega^2 \left( \sigma_1^2(\theta_1,\theta_2) -2\, \sigma_2(\theta_1,\theta_2) \right),
\\ \nonumber \\
\tilde{Q}^S_{-,1} & = & \sigma_3(\chi_1,\chi_2,\chi_3,\chi_4) \, \omega + \sigma_4(\chi_1,\chi_2,\chi_3,\chi_4)  \\ \nonumber
& = & \omega  \left(\chi_1 \chi_2 \chi_3+\chi_1 \chi_4 \chi_3+\chi_2 \chi_4 \chi_3+\chi_1 \chi_2 \chi_4\right)+\chi_1 \chi_2 \chi_3 \chi_4 \\ \nonumber
& = & -4 \,  \sigma_1^2(\theta_1,\theta_2)\, \sigma_2(\theta_1,\theta_2) + \sigma_1^4(\theta_1,\theta_2),
\\ \nonumber \\
\tilde{Q}^S_{+,2} & = & \sigma_1(\chi_1,\chi_2,\chi_3,\chi_4) \, \omega^{3} + \sigma_3(\chi_1,\chi_2,\chi_3,\chi_4) \, \omega \\ \nonumber
& = & \omega^3 \left(\chi_1+\chi_2+\chi_3+\chi_4\right)+\omega  \left(\chi_1 \chi_2 \chi_3+\chi_1 \chi_4 \chi_3+\chi_2 \chi_4 \chi_3+\chi_1 \chi_2 \chi_4\right) \\ \nonumber
& = & 0,
\\ \nonumber \\
\tilde{Q}^S_{-,2} & = & \sigma_2(\chi_1,\chi_2,\chi_3,\chi_4) \, \omega^{2} + \sigma_4(\chi_1,\chi_2,\chi_3,\chi_4) \\ \nonumber
& = & \omega^2 \left(\chi_1 \chi_2+\chi_3 \chi_2+\chi_4 \chi_2+\chi_1 \chi_3+\chi_1 \chi_4+\chi_3 \chi_4\right)+\chi_1 \chi_2 \chi_3 \chi_4 \\ \nonumber
& = &  \sigma_1^4(\theta_1,\theta_2) - \sigma_1^2(\theta_1,\theta_2)\, \left( 2 \, \omega^2 + 4 \, \sigma_2(\theta_1,\theta_2) \right) + 4 \, \omega^2 \, \sigma_2(\theta_1,\theta_2).
\end{eqnarray*}
\item Therefore, if $\tilde{Q} \subset \mathbb{C}[\theta_1,\theta_2,\omega]$ denotes the ideal generated by $\tilde{Q}^S_{\pm,1}, \tilde{Q}^S_{\pm,2}$ in Eqs. \eqref{kernel-generators-r=2-eq1}--\eqref{kernel-generators-r=2-eq4}, then we will have
\begin{equation}
H^*(M_{\xi}) = \mathbb{C}[\theta_1,\theta_2,\omega]/\tilde{Q}, \quad \xi \in \Delta_{\text{reg}},
\end{equation}
where
\[
\tilde{Q} = \sum_{l=1}^{2}{\tilde{Q}^S_{+,l} \oplus \tilde{Q}^S_{-,l}}.
\]
\end{itemize}

In the next example we consider the case $r=3$, namely the torus $T^3=S \times S \times S$ acting in a Hamiltonian fashion on the K\"ahler manifold $M=\mathbb{P}_7$ via the homomorphism \eqref{action}. We can find generators of the kernel of surjective homomorphism $\kappa \circ h: \mathbb{C}[\theta_1,\theta_2,\theta_3,\omega] \rightarrow H^*(M_{\xi})$ as follows:
\begin{itemize}
\item Recall that the matrix $a_{i,j}\equiv A$ is given by
\begin{equation} \label{a-matrix-r=3}
A \equiv a_{i,j} = \left( \begin{matrix}
1 & 1 & 1 & -1 & 1 & -1 & -1 & -1\\
1 & 1 & -1 & 1 & -1 & 1 & -1 & -1 \\
1 & -1 & 1 & 1 & -1 & -1 & 1 & -1
\end{matrix}\right)_{3 \times 8}.
\end{equation}
\item The moment polytope is a cube (3-cube) with 8 vertices as the image under the moment map $\mu_T : \mathbb{P}_7 \rightarrow \mathfrak{t}^*$. The symplectic quotient $M_{\xi} = \mu_T^{-1}(\xi)/T^3$, with $\xi \in \Delta_{\text{reg}}$ as shown in Figure. \ref{3-cube-example}, possesses at most orbifold singularities. 
\item Recalling the corollary \ref{torus-kalkman-cor} and Eq. \eqref{a-matrix-r=3}, for each 1-dimensional torus (circle) $S$ we have
\begin{itemize}
\item For first circle $S$ action with $\alpha_1 \in \mathfrak{s}^*$:
\begin{eqnarray} \label{kernel-generators-r=3-eq1}
\tilde{Q}^S_{+,1} & = & \prod_{\substack{p \in \mathcal{F} \\ \langle \mu_S(p), \alpha_1 \rangle > \langle \xi,\alpha_1 \rangle}}{ \left( \omega + \sum_{i=1}^{3}{a_{i,j}\theta_i} \right) } \\ \nonumber 
& = & (\omega + \theta_1 + \theta_2 + \theta_3)(\omega + \theta_1 + \theta_2 - \theta_3)(\omega + \theta_1 - \theta_2 + \theta_3)(\omega + \theta_1 - \theta_2 - \theta_3),
\end{eqnarray}
\begin{eqnarray} \label{kernel-generators-r=3-eq2}
\tilde{Q}^S_{-,1} & = & \prod_{\substack{p \in \mathcal{F} \\ \langle \mu_S(p), \alpha_1 \rangle < \langle \xi,\alpha_1 \rangle}}{ \left( \omega + \sum_{i=1}^{3}{a_{i,j}\theta_i} \right) } \\ \nonumber 
& = & (\omega - \theta_1 + \theta_2 + \theta_3)(\omega - \theta_1 + \theta_2 - \theta_3)(\omega - \theta_1 - \theta_2 + \theta_3)(\omega - \theta_1 - \theta_2 - \theta_3).
\end{eqnarray}

\item For second circle $S$ action with $\alpha_2 \in \mathfrak{s}^*$:
\begin{eqnarray} \label{kernel-generators-r=3-eq3}
\tilde{Q}^S_{+,2} & = & \prod_{\substack{p \in \mathcal{F} \\ \langle \mu_S(p), \alpha_2 \rangle > \langle \xi,\alpha_2 \rangle}}{ \left( \omega + \sum_{i=1}^{3}{a_{i,j}\theta_i} \right) } \\ \nonumber 
& = & (\omega + \theta_1 + \theta_2 + \theta_3)(\omega + \theta_1 + \theta_2 - \theta_3)(\omega - \theta_1 + \theta_2 + \theta_3)(\omega - \theta_1 + \theta_2 - \theta_3),
\end{eqnarray}
\begin{eqnarray} \label{kernel-generators-r=3-eq4}
\tilde{Q}^S_{-,2} & = & \prod_{\substack{p \in \mathcal{F} \\ \langle \mu_S(p), \alpha_2 \rangle < \langle \xi,\alpha_2 \rangle}}{ \left( \omega + \sum_{i=1}^{3}{a_{i,j}\theta_i} \right) } \\ \nonumber 
& = & (\omega + \theta_1 - \theta_2 + \theta_3)(\omega + \theta_1 - \theta_2 - \theta_3)(\omega - \theta_1 - \theta_2 + \theta_3)(\omega - \theta_1 - \theta_2 - \theta_3).
\end{eqnarray}

\item For third circle $S$ action with $\alpha_3 \in \mathfrak{s}^*$:
\begin{eqnarray} \label{kernel-generators-r=3-eq5}
\tilde{Q}^S_{+,3} & = & \prod_{\substack{p \in \mathcal{F} \\ \langle \mu_S(p), \alpha_3 \rangle > \langle \xi,\alpha_3 \rangle}}{ \left( \omega + \sum_{i=1}^{3}{a_{i,j}\theta_i} \right) } \\ \nonumber 
& = & (\omega + \theta_1 + \theta_2 + \theta_3)(\omega + \theta_1 - \theta_2 + \theta_3)(\omega - \theta_1 + \theta_2 + \theta_3)(\omega - \theta_1 - \theta_2 + \theta_3),
\end{eqnarray}
\begin{eqnarray} \label{kernel-generators-r=3-eq6}
\tilde{Q}^S_{-,3} & = & \prod_{\substack{p \in \mathcal{F} \\ \langle \mu_S(p), \alpha_3 \rangle < \langle \xi,\alpha_3 \rangle}}{ \left( \omega + \sum_{i=1}^{3}{a_{i,j}\theta_i} \right) } \\ \nonumber 
& = & (\omega + \theta_1 + \theta_2 - \theta_3)(\omega + \theta_1 - \theta_2 - \theta_3)(\omega - \theta_1 + \theta_2 - \theta_3)(\omega - \theta_1 - \theta_2 - \theta_3).
\end{eqnarray}
\end{itemize}

\item In terms of elementary symmetric polynomials $\sigma_j$ in Eq. \eqref{symmetric-pol-eq1} and recalling the fact that $\chi_1=\theta_1 + \theta_2 + \theta_3, \, \chi_2 = \theta_1 + \theta_2 - \theta_3, \, \chi_3 = \theta_1 - \theta_2 + \theta_3, \, \chi_4 =- \theta_1 + \theta_2 + \theta_3, \chi_5=\theta_1 -\theta_2 - \theta_3, \chi_6 = -\theta_1 + \theta_2 -\theta_3 ,\chi_7=-\theta_1 - \theta_2 +\theta_3, \chi_8 = -\theta_1 - \theta_2 - \theta_3$, we have
\begin{eqnarray*}
\tilde{Q}^S_{+,1} & = & \sigma_1(\chi_{\bullet}) \, \omega^{7} + \sigma_2(\chi_{\bullet}) \, \omega^{6} + \sigma_3(\chi_{\bullet}) \, \omega^{5} + \sigma_5(\chi_{\bullet}) \, \omega^{3} \\ \nonumber
& = & -4 \, \omega^6  \sigma_1^2(\theta_{\bullet}) + 8 \, \omega^6 \sigma_2(\theta_{\bullet}),
\\ \nonumber \\
\tilde{Q}^S_{-,1} & = & \sigma_4(\chi_{\bullet}) \, \omega^{4} + \sigma_6(\chi_{\bullet}) \, \omega^{2} + \sigma_7(\chi_{\bullet}) \, \omega^{1} + \sigma_8(\chi_{\bullet}) \, \omega^{0} \\ \nonumber
& = & -4 \sigma_1^6(\theta_{\bullet}) \left(2 \sigma_2(\theta_{\bullet})+\omega^2\right)-32 \sigma_3(\theta_{\bullet}) \sigma_1^3(\theta_{\bullet}) \left(2 \sigma_2(\theta_{\bullet})+\omega^2\right)+ \\ \nonumber 
& & 16 \sigma_3(\theta_{\bullet}) \sigma_1(\theta_{\bullet}) \omega^2 \left(4 \sigma_2(\theta_{\bullet})+\omega^2\right)+ 16 \omega^2 \left(\sigma_2^2(\theta_{\bullet}) \omega^2-4 \sigma_3^2(\theta_{\bullet})\right)+ \\ \nonumber
& & 2 \sigma_1^4(\theta_{\bullet}) \left(12 \sigma_2(\theta_{\bullet}) \omega^2+8 \sigma_2^2(\theta_{\bullet})+3 \omega^4\right)-8 \sigma_1^2(\theta_{\bullet}) \left(3 \sigma_2(\theta_{\bullet}) \omega^4+4 \sigma_2^2(\theta_{\bullet}) \omega^2-8 \sigma_3^2(\theta_{\bullet})\right)+ \\ \nonumber 
& & \sigma_1^8(\theta_{\bullet})+16 \sigma_3(\theta_{\bullet}) \sigma_1^5(\theta_{\bullet}),
\\ \nonumber \\
\tilde{Q}^S_{+,2} & = & \sigma_1(\chi_{\bullet}) \, \omega^{7} + \sigma_2(\chi_{\bullet}) \, \omega^{6} + \sigma_4(\chi_{\bullet}) \, \omega^{4} + \sigma_6(\chi_{\bullet}) \, \omega^{2} \\ \nonumber
& = & 2 \omega^2 ( 3 \sigma_1^4(\theta_{\bullet}) \left(4 \sigma_2(\theta_{\bullet})+\omega^2\right)+ 8 \sigma_3(\theta_{\bullet}) \sigma_1(\theta_{\bullet}) \left(4 \sigma_2(\theta_{\bullet})+\omega^2\right)- \nonumber \\
& & 2 \sigma_1^2(\theta_{\bullet}) \left(6 \sigma_2(\theta_{\bullet}) \omega^2+8 \sigma_2^2+\omega^4\right)+ 4 \left(\sigma_2(\theta_{\bullet}) \omega^4+2 \sigma_2^2(\theta_{\bullet}) \omega^2-8 \sigma_3^2(\theta_{\bullet})\right)- \\ \nonumber 
& & 2 \sigma_1^6(\theta_{\bullet})-16 \sigma_3(\theta_{\bullet}) \sigma_1^3(\theta_{\bullet}) ),
\\ \nonumber \\
\tilde{Q}^S_{-,2} & = & \sigma_3(\chi_{\bullet}) \, \omega^{5} + \sigma_5(\chi_{\bullet}) \, \omega^{3} + \sigma_7(\chi_{\bullet}) \, \omega^{1} + \sigma_8(\chi_{\bullet}) \, \omega^{0} \\ \nonumber
& = & \sigma_1^2(\theta_{\bullet}) \left(\sigma_1^3(\theta_{\bullet})-4 \sigma_2(\theta_{\bullet}) \sigma_1(\theta_{\bullet})+8 \sigma_3(\theta_{\bullet})\right)^2,
\\ \nonumber \\
\tilde{Q}^S_{+,3} & = & \sigma_1(\chi_{\bullet}) \, \omega^{7} + \sigma_3(\chi_{\bullet}) \, \omega^{5} + \sigma_4(\chi_{\bullet}) \, \omega^{4} + \sigma_7(\chi_{\bullet}) \, \omega^{1} \\ \nonumber
& = & 2 \left(3 \sigma_1^4(\theta_{\bullet})-12 \sigma_2(\theta_{\bullet}) \sigma_1^2(\theta_{\bullet})+8 \sigma_3(\theta_{\bullet}) \sigma_1(\theta_{\bullet})+8 \sigma_2^2(\theta_{\bullet})\right) \omega^4,
\\ \nonumber \\
\tilde{Q}^S_{-,3} & = & \sigma_2(\chi_{\bullet}) \, \omega^{6} + \sigma_5(\chi_{\bullet}) \, \omega^{3} + \sigma_6(\chi_{\bullet}) \, \omega^{2} + \sigma_8(\chi_{\bullet}) \, \omega^{0} \\ \nonumber
& = & -4 \sigma_1^6(\theta_{\bullet}) \left(2 \sigma_2(\theta_{\bullet})+\omega^2\right)+8 \sigma_2(\theta_{\bullet}) \sigma_1^4(\theta_{\bullet}) \left(2 \sigma_2(\theta_{\bullet})+3 \omega^2\right)- \\ \nonumber 
& & 32 \sigma_3(\theta_{\bullet}) \sigma_1^3(\theta_{\bullet}) \left(2 \sigma_2(\theta_{\bullet})+\omega^2\right)+64 \sigma_2(\theta_{\bullet}) \sigma_3(\theta_{\bullet}) \sigma_1(\theta_{\bullet}) \omega^2- \\ \nonumber 
& & 4 \sigma_1^2(\theta_{\bullet}) \left(8 \sigma_2^2(\theta_{\bullet}) \omega^2-16 \sigma_3^2(\theta_{\bullet})+\omega^6\right)+8 \omega^2 \left(\sigma_2(\theta_{\bullet}) \omega^4-8 \sigma_3^2(\theta_{\bullet})\right)+ \\ \nonumber 
& & \sigma_1^8(\theta_{\bullet})+16 \sigma_3(\theta_{\bullet}) \sigma_1^5(\theta_{\bullet}).
\end{eqnarray*}

\item Therefore, with $\tilde{Q} \subset \mathbb{C}[\theta_1,\theta_2,\theta_3,\omega]$ being the ideal generated by $\tilde{Q}^S_{\pm,1}, \tilde{Q}^S_{\pm,2}, \tilde{Q}^S_{\pm,3}$ in Eqs. \eqref{kernel-generators-r=3-eq1}--\eqref{kernel-generators-r=3-eq6}, then we have
\begin{equation}
H^*(M_{\xi}) = \mathbb{C}[\theta_1,\theta_2,\theta_3,\omega]/\tilde{Q}, \quad \xi \in \Delta_{\text{reg}},
\end{equation}
where
\[
\tilde{Q} = \sum_{l=1}^{3}{\tilde{Q}^S_{+,l} \oplus \tilde{Q}^S_{-,l}}.
\]
\end{itemize}

\section{Recursive Wall-crossing Integrals} \label{recursive-wall-crossing-integrals}

\subsection{The Localization Algorithm} \label{localization-sub-sec}
The main goal of this section is to obtain a recursive wall-crossing localization algorithm to compute integrals of a cohomology class $\kappa (a) \in H^*(M_{\xi})$ over the abelian symplectic quotients $M_{\xi}$, where $\kappa:H^*_T(M) \rightarrow H^*(M_{\xi})$ is the Kirwan map \eqref{kirwan-map-eq} with an equivariant cohomology class $a \in H^*_T(M)$, as $\xi \in \Delta_{\text{reg}}$ crosses critical walls of the moment polytope $\Delta \equiv \mu_T(M)$. Using the localization techniques to compute such integrals have been studied before, for instance \cite{guillemin1996,jeffrey1997,martin2000a}. In this section, we will apply the following recursive wall-crossing formula \cite[Theorem 3.1]{guillemin1996}
\begin{equation} \label{integral-wall-crossing-eq}
\int_{M_{\xi}} \kappa(a) = \sum^{\prime}_{i} \int_{M^{(i)}_{r-1}} \kappa_{r-1}[(\text{res})_{(r-1)}(a)],
\end{equation}
where $a \in H^{2m}_T(M)$ with $m = \text{dim}_{\mathbb{C}}(M_{\xi})$ and $M^{(i)}_{r-1} = (\mu^{-1}_{T^{r-1}}(q^{(i)}_{r-1}) \, \cap \, F^i_{r-1})/(T^{r-1})$ and an ``$i$'' occurs in the summation on the right hand side of Eq. \eqref{integral-wall-crossing-eq} if and only if $F_{r-1}^i \, \cap \, \mu_T^{-1}(l_{r-1}(\xi))$ intersects co-dimension-$1$ critical walls of the moment polytope $\Delta$ at the points $q^{(i)}_{r-1}$. Here, $F_{r-1}^i$ denotes the connected components of fixed point set of a one-parameter subgroup $T^1$ of $T^r$ and $l_{r-1}(\xi)$ is an appropriately chosen ray through $\xi \in \Delta_{\text{reg}}$ and along a non-zero element of the weight lattice of $\mathfrak{t}^*_1$. In fact, $\mu_T(F^i_{r-1})$ form co-dimension-$1$ walls of the moment polytope $\Delta$ and $F^i_{r-1}$ is a symplectic manifold under the effective and Hamiltonian action of the quotient torus $T^r/T^1 \equiv T^{r-1}$. Then, the Kirwan map for the sub-reductions are $\kappa_{r-1}: H^*_{T^{r-1}}(F^i_{r-1}) \rightarrow H^*(M_{r-1}^{(i)})$ and the corresponding residue operation in Eq. \eqref{integral-wall-crossing-eq} is the following map
\begin{equation} \label{residue-map-eq}
(\text{res})_{(r-1)}: H^*_{T^r}(M) \rightarrow H^*_{T^{r-1}}(F^i_{r-1}), (\text{res})_{(r-1)}(a) \mapsto (\text{res})_{\theta_1=0}(\imath^*_F a),
\end{equation}
where $\imath^*_F(a) \in H^*_{T^{r}}(F^i_{r-1}) = H^*_{T^{r-1}}(F^i_{r-1}) \otimes H^*_{T^1}$ and $\theta_1$ is a basis element for $\mathfrak{t}^*_1 \cong \mathbb{R}$. Recall that the infinitesimal orbit-type stratification \cite{saeid2015a} decomposes the symplectic manifold $M$ into the connected components $F_{r-k}$ of the fixed point sets of the sub-tori $T^k$ actions, for $k=0, \cdots, r$. More precisely, we have 
\begin{equation}
F \equiv F_0 \subset F_1 \subset \cdots \subset F_{r-1} \subset F_r \cong M,
\end{equation}
where each submanifold $F_{r-k}$ itself is a symplectic manifold under the Hamiltonian action of the quotient torus $T^r/T^k \cong T^{r-k}$ and is equipped with an equivariant moment map $\mu_{r-k}:F_{r-k} \rightarrow \mathfrak{t}^*_{r-k}$, where $\mathfrak{t}^*_{r-k}$ is the dual to the Lie algebra $\mathfrak{t}_{r-k}$ of the torus $T^{r-k}$. It is important to note that the set of regular values of the moment map $\mu_{T^r}$ is the complement of the union of the critical walls of the moment polytope $\mu_{T^r}(M)= \Delta$, or
\[
\Delta_{\text{reg}} = \Delta \backslash \bigcup_s \mu_{s}(F_s).
\]
In \cite{guillemin1996}, Guillemin and Kalkman used a certain combinatorial object called ``dendrite'' to iteratively express $\int_{M_\xi}\kappa(a)$ in terms of the integration over the connected components of $T^r$ fixed point set $M^{T^r} \cong F$. In particular, the recursion starts by considering the symplectic quotient $M_{\xi}$ at a desired point $\xi \in \Delta_{\text{reg}}$ and follows by a finite collection $D$ of certain tuples $(l_s(\xi),W_s)$, for $s=0, \cdots, r-1$, where $W_s = \mu_{s}(F_s)$ is a sub-polytope, or a critical wall of the moment polytope $\Delta$ and $l_s(\xi)$ is a ray (one dimensional shifted cone, i.e. $l_s(\xi) = \left\{ \xi + \sum_{s} t_s \, \theta_s \, | \, 0 \leq t_s < \infty \right\}$, where $\left\{ \theta_s \right\}$ forms a basis for the weight lattice of $\mathfrak{t}^*_{s}$), whose intersection with $\Delta \cap W_s$ give rises to a finite collection of points $q_s$. Then $D$ is called a dendrite, if the following conditions hold \cite{jeffrey2005b}:
\begin{enumerate}
	\item For $(l_s(\xi),W_s=\mu_{s}(F_s))$, the ray $l_s(\xi)$ is transverse to the $\mu_{s}(F_s)$,
	\item For $(l_s(\xi),W_s=\mu_{s}(F_s))$, if $l_s(\xi)$ intersects a co-dimension one wall $W^{\prime}_s$ of the $\mu_{s}(F_s)$ at a point $q^{\prime}_s$, then there is a unique cone $l^{\prime}_s(\xi)$ such that $(l^{\prime}_s(\xi),W^{\prime}_s) \in D$.
\end{enumerate}
Therefore, a dendrite will consist of a main branch (the ray $l_{r-1}(\xi)$ through $\xi \in \Delta_{\text{reg}}$ and $q_{r-1} \in \Delta \cap W_{r-1}$), the secondary branches (the rays $l^k_{r-2}(\xi)$ through $q_{r-1}$ and $q_{r-2} \in \Delta \cap W_{r-2}$), \dots until it terminates at the vertices, or the $\mu_T(F)$. Moreover, for every Hamiltonian submanifold $(F_s,\mu_{s},T^{s})$, there exists a moment sub-polytope $\mu_{s}(F_s) \subset \Delta$, whose set of regular values $\Delta^{(s)}_{\text{reg}}$ are given by the complement of the union of the critical sub-walls as follows
\[
\Delta^{(s)}_{\text{reg}} = \mu_{s}(F_s) \backslash \bigcup_{F_{s^{\prime}}<F_s} \mu_{s^{\prime}}(F_{s^{\prime}}),
\]
which itself is a relatively open set with finite number of components \cite{saeid2015a}. In other words, for each of the points $q_s \in \Delta^{(s)}_{\text{reg}}$ we can define a regular symplectic quotient defined by the following sub-reductions
\begin{equation} \label{symplectic-subreduction-eq}
M_{q_s} = \left( \mu^{-1}_{s}(q_s) \cap F_s \right)/ (T^{s}), \quad s = 0 , \cdots, r-1. 
\end{equation}
Now, consider the formula \eqref{integral-wall-crossing-eq}, in which every term in the summation in the right-hand side corresponds to the symplectic quotient of the kind $M_{q_{r-1}} \equiv M^{(i)}_{r-1}$, for which repeating the same procedure would lead to the following
\begin{equation} \label{symplectic-subreduction-eq1}
\int_{M^{(i)}_{r-1}} \kappa_{(r-1)}[(\text{res})_{(r-1)}(a)] = \sum_{i^{\prime}}\int_{M^{(i^{\prime})}_{r-2}} \kappa_{(r-2)}[(\text{res})_{(r-2)}(a)],
\end{equation}  
where the $i^{\prime}$ occurs in the summation if and only if $\mu^{-1}_T(l_{r-2}(\xi)) \cap F_{r-2}^{i^{\prime}}$ intersects the co-dimension one critical walls of the moment subpolytope $\Delta^{(r-1)}$ in the points $q_{r-2}^{i^{\prime}}$ and also $M^{(i^{\prime})}_{r-2} \equiv M_{q_{r-2}}$. Obviously, we can repeat the same procedure for each term in the right-hand side of the Eq. \eqref{symplectic-subreduction-eq1}, i.e. for integrals over symplectic quotients in terms of sub-reductions at certain points that are determined by the dendrite. The recursion will terminate at the fixed point set $F \equiv M^{T^r}$ of the $T^r$-action, or the vertices of the moment polytope $\Delta$. More specifically, a dendrite $D$ consists of certain paths $P$, given by a sequence of tuples
\begin{equation} \label{dendrite-path-eq}
\left( (l_{r-1}(\xi),W_{r-1}), \, (l_{r-2}(\xi),W_{r-2}), \, \cdots , (l_0(\xi),W_0 \in F) \right).
\end{equation}

Let $a \in H^*_T(M)$ and for an appropriately chosen set of coordinates $\left\{ \theta_{s+1} \right\}$ on the Lie algebra of $T^{s+1}$, for $s=0, \cdots , r-1$, define
\begin{equation} \label{residue-eq1}
R_s(a) = \text{res}_{\theta_{s+1}=0} \left( \frac{ \imath^*_{F}(a)}{e_T(\nu_s)} \right), \quad s=0, \cdots , r-1,
\end{equation}
where $\imath^*_{F}$ is the natural inclusion $F \equiv M^{T^r} \hookrightarrow M$ and $\nu_s$ is the restriction to $F$ of the normal bundle to $F_s$ and hence $e_T(\nu_s)$ the corresponding equivariant Euler class of $\nu_s$ defined by
\begin{equation} \label{equivariant-euler-class-general-eq}
e_T(\nu_s) =\prod_{k=0}^{s}e_T(\nu_k),
\end{equation}
where the normal bundle $\nu \equiv \nu_0$ of $F$ in $M$ is the sum $\nu = \bigoplus_s \nu_s$, according to the splitting principle \cite{bott1982}. 
Therefore, the contribution $C_P$ to $\int_{M_{\xi}}{\kappa(a)}$ for a path $P$ is given by the following localization formula
\begin{equation} \label{iterated-residues-general-eq}
C_P(\imath^*_{F}(a)) = R_{0} \circ R_{1} \circ \cdots \circ R_{r-1}(\imath^*_{F}(a)).
\end{equation}
\begin{thm} \label{iterated-residues-thm}
\cite{guillemin1996,jeffrey2005b} For a dendrite $D$ and $a \in H^*_T(M)$,
\begin{equation}
\int_{M_{\xi}}{\kappa(a)} = \sum_{P}{\int_{F_P}{C_P(\imath^*_{F_P}(a))}},
\end{equation}
where the above sum is over all possible paths $P$ belonging to the dendrite $D$.
\end{thm}
In the next sub-section, we will use the recursive procedure and the iterated residues in the Theorem \ref{iterated-residues-thm} in our case.

\subsection{Iterated Residues for Abelian Quotients} \label{iterated-residue-alg}
In this section, we will use the methods of iterated residues of the localization introduced in the previous sub-section \ref{localization-sub-sec} to the torus action of the section \ref{introduction}. Recall that the torus $T^r$ action on the manifold $M=\mathbb{P}_n$, with $2^r=n+1$, is given by the homomorphism $\varphi: T^r \rightarrow U(n+1)$ of Eq. \eqref{action} as follows
\begin{equation} \label{action-eq}
(t_1,t_2, \cdots , t_r) \mapsto \left( \prod_{i=1}^{r}t_i^{a_{i,1}}, \prod_{i=1}^{r}t_i^{a_{i,2}}, \cdots, \prod_{i=1}^{r}t_i^{a_{i,n+1}}\right),
\end{equation}
for some $a_{i,j} \in \left\{ \pm 1 \right\}$, where $\sum_{k=0}^{r}{\binom{r}{k}}=2^r=n+1$. For fixed points $p_j \in F \equiv M^{T^r}$, the \textit{isotropy weights} $w_{p_j}$ of the infinitesimal $T^r$ action on $T_pM$ are obtained by 
\begin{equation} \label{isotropy-weights-0-eq}
w_{p_j} \equiv w_{l,j} = \sum_{k=1}^{r}{(a_{k,l} - a_{k,j}) \, \theta_k} = \sum_{k=1}^{r}{\beta_k \, \theta_k}, \quad l \neq j,
\end{equation}
where $\left\{ \theta_k \right\}_{k=1}^{r}$ form a basis for the dual of the Lie algebra $\mathfrak{t}^*_r$. We can always fix an element $\gamma \in \mathfrak{t}_r$, which is not orthogonal to all the isotropy weights $\left\{ w_{l,j} \right\}$, for all the fixed points $p_j$ with $j=1, \cdots , n+1$, namely $\langle w_{l,j}, \gamma \rangle \equiv w_{l,j}(\gamma) < 0$ for $l=1, \cdots \hat{j}, \cdots , d$, and $w_{l,j}(\gamma) > 0$ for $l=d+1, \cdots \hat{j}, \cdots , r$, where $\hat{j}$ denotes $l \neq j$. Let $\epsilon_1 = \cdots = \epsilon_d=-1$ and $\epsilon_{d+1} = \cdots = \epsilon_r=1$ and set $\tilde{w}_{l,j} = \epsilon_l \, w_{l,j}$, for $l = 1, \cdots , \hat{j}, \cdots , n+1$. This procedure is called \textit{polarization} and the resulting weights $\tilde{w}_{l,j}$ are known as the \textit{polarized weights} \cite{ginzburg2002,guillemin2003,christandl2014}.

Moreover, at every fixed point $p_j \in F$, we have $T_{p_j}\nu_0 \cong T_{p_j}M$, which splits into a direct sum of $T$-equivariant line bundles, i.e. $T_{p_j}\nu_0 \cong \oplus_{l=1}^{n} L_l^{(0)}$, for each of which the polarized version of the $T$-equivariant first Chern class $c^T_1(L_l^{(0)})$ is given by the polarized isotropy weights $\tilde{w}_{l,j} \equiv \tilde{w}_{l,j}^{(0)}$, such that $l \neq j$. We also have to note that there are exactly $n$ (polarized) isotropy weights $\tilde{w}_{l,j}^{(0)}$ at each fixed point $p_j$ and not any random $(r-1)$-tuple of them satisfy the \textit{genericity} condition of \cite{guillemin1996}, i.e. not any random $(r-1)$-tuple of isotropy weights at every fixed point are linearly independent. Hence, the $T$-equivariant Euler class $e_T(\nu_0)$ is given by
\begin{equation} \label{equivariant-euler-class-normal-0-eq}
e_{T^r}(\nu_0) = \prod_{\substack{l \neq j}}^{2^r-1}{\tilde{w}_{l,j}^{(0)}} = \prod_{\substack{l \neq j}}^{2^r-1}{ \left( \epsilon_l \sum_{k=1}^{r}{(a_{k,l} - a_{k,j}) \theta_k} \right) }.
\end{equation}
The polarized isotropy weights $\tilde{w}_{l,j}^{(s)}$ for the infinitesimal action of $T^r$ on the normal bundle $\nu_s$ is given by
\begin{equation} \label{isotropy-weights-s-eq}
\tilde{w}_{l,j}^{(s)}  = \sum_{k=s+1}^{r}{\tilde{\beta}_k^{(s)} \theta_k} = \sum_{k=s+1}^{r}{\epsilon_l(a_{k,l} - a_{k,j}) \, \theta_k} \in \mathfrak{t}^*_{r-s}, \quad l \neq j,
\end{equation}
for $s=1, \cdots , r-1$, since at any point $p \in M$, we have the decomposition $T_pM \cong T_pF_s \oplus T_p\nu_s$ and the moment images $\mu_T(M) \in \mathfrak{t}^*_r$ and $\mu_T(F_s) \in \mathfrak{t}^*_s$ and so $\mu_T(\nu_s) \in \mathfrak{t}^*_{r-s}$. Recall that every sub-manifold $F_s$ is a Hamiltonian manifold under the action of the sub-torus $T^s$, given by the following homomorphism
\begin{equation}
\varphi^{\prime}: T^s \rightarrow U(n^{\prime}+1): \left( t_1, \cdots, t_s \right) \mapsto \left( \prod_{i^{\prime}=1}^{s} t^{a^{\prime}_{i^{\prime},1}}, \cdots , \prod_{i^{\prime}=1}^{s} t^{a^{\prime}_{i^{\prime},n^{\prime}+1}} \right),
\end{equation}
with $a^{\prime}_{i,j} \in \left\{ \pm 1\right\}$, where $n^{\prime} +1= 2^s$. Hence, by using a simple dimensional analysis we find that $\text{dim}(T_p(\nu_s)) = \text{dim}(T_pM) - \text{dim}(T_p(F_s)) = 2^r-2^s$. Therefore, the equivariant Euler class $e_{T^r}(\nu_s)$ to the normal bundle $\nu_s$ of the submanifold $F_s$ can be obtained as follows
\begin{equation} \label{equivariant-euler-class-normal-s-eq}
e_{T^r}(\nu_s) = \prod_{\substack{l \neq j}}^{2^r-2^s}{\tilde{w}_{l,j}^{(s)}} = \prod_{\substack{l \neq j}}^{2^r-2^s}{ \left( \sum_{k=1}^{r}{\epsilon_l(a_{k,l} - a_{k,j}) \theta_k} \right) },
\end{equation}
with a consistent choice of orientations at all the stages of this iteration, namely for $s=0, \cdots, r-1$.
 
\subsection{Examples} 
\subsubsection{Two Qubits} \label{two-qubits-example-integral}
Let's consider again the case $r=2$, in which the torus $T^2=S^1 \times S^1$ acts in a Hamiltonian fashion on the K\"ahler manifold $\mathbb{P}_3$ through the homomorphism \eqref{action}. Recall that the matrix $a_{i,j}\equiv A$ is given by
\begin{equation} \label{a-matrix-r=2II}
A \equiv a_{i,j} = \left( \begin{matrix}
1 & 1 & -1 & -1 \\
1 & -1 & 1 & -1
\end{matrix}\right)_{2 \times 4}.
\end{equation}
The moment polytope $\Delta$ is a square (2-cube) with 4 vertices as the image under the moment map $\mu_T : \mathbb{P}_3 \rightarrow \mathfrak{t}^*$ and the symplectic quotient $M_{\xi} = \mu_T^{-1}(\xi)/T^2$, with $\xi \in \Delta_{\text{reg}}$ as shown in Figure. \ref{2-cube}, possesses at most orbifold singularities.
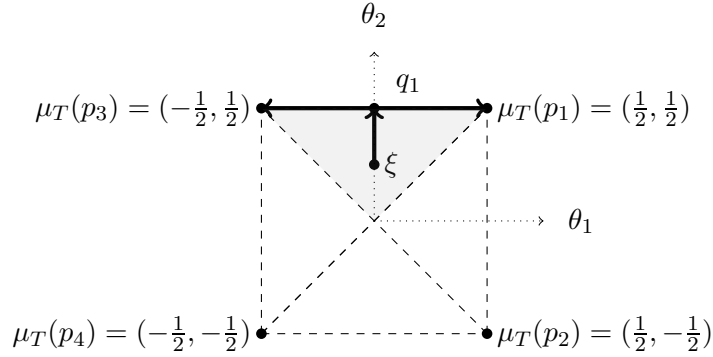
\begin{figure}[h]
\centering
\begin{tikzpicture}[
  vector/.style={ultra thick,black,>=stealth,->},
  atom/.style={blue}, x=3cm,y=3cm
  ]
    \coordinate (p0) at (0,0);
    \fill[black] (p0) circle (2pt);
    \node[left] at (p0) {$\mu_T(p_4) = (-\frac{1}{2},-\frac{1}{2})$};
    \coordinate (p1) at (1,0);
    \fill[black] (p1) circle (2pt);
    \node[right] at (p1) {$\mu_T(p_2) = (\frac{1}{2},-\frac{1}{2})$};
    \coordinate (p2) at (0,1);
    \fill[black] (p2) circle (2pt);
    \node[left] at (p2) {$\mu_T(p_3) = (-\frac{1}{2},\frac{1}{2})$};
    \coordinate (p3) at (1,1);
    \fill[black] (p3) circle (2pt);
    \node[right] at (p3) {$\mu_T(p_1) = (\frac{1}{2},\frac{1}{2})$};
    
    \coordinate (a1) at (.5,.5);
    \coordinate (a6) at (.5,1.25);
    \node[above=.5em of a6] {$\theta_2$};
    \coordinate (a5) at (1.25,.5);
    \node[right=.5em of a5] {$\theta_1$};
    
    \coordinate (a2) at (0.5,1);
    \node[above=.5em of a2] (dummy) {};
    \node[right=.05em of dummy] {$q_1$};
    \fill[black] (a2) circle (2pt);    
    \draw[ultra thick,->] (a2) -- (p3);
    \draw[ultra thick,->] (a2) -- (p2);
    
     \coordinate (a4) at (0.5,0.75);
     \node[right] at (a4) {$\xi$};
     \fill[black] (a4) circle (2pt);
     \draw[ultra thick,->] (a4) -- (a2);
        
\fill[gray,opacity=0.1] (p2) -- (p3) -- (a1);
    \draw[dotted,->] (a1) -- (a5);
    \draw[dotted,->] (a1) -- (a6);
    \draw[dashed] (p0) -- (p1);
    \draw[dashed] (p0) -- (p2);
    \draw[dashed] (p0) -- (p3);
    \draw[dashed] (p0) -- (p3);
    \draw[dashed] (p1) -- (p2);
    \draw[dashed] (p3) -- (p2);
    \draw[dashed] (p1) -- (p3);
    \draw[dashed] (p2) -- (p3);
\end{tikzpicture}
\caption{$\mu_T(\mathbb{P}_3)\equiv \Delta_2$, for $M = \mathbb{P}_3$ with $T^2$ Hamiltonian action. The gray area represent the upper $\Delta_{\text{reg}}$, to which $\xi$ belongs and the black arrows illustrates the corresponding dendrite.}
\label{2-cube} 
\end{figure}
Recall that $\text{dim}(M_{\xi})=2$ and the goal is to find $\int_{M_{\xi}}{\kappa(\tilde{\omega})}$, where $\tilde{\omega} \in H^2_{T^2}(\mathbb{P}_3;\mathbb{C})$, such that $\kappa(\tilde{\omega}) = \omega_{\xi}$ is the reduced symplectic form on $M_{\xi}$, i.e. $\omega_{\xi} \in H^2(M_{\xi})$. Recall again that according to the Duistermaat-Heckman theorem \cite{dh1982}, if $\eta,\xi \in \Delta$ belong to the same chamber of regular values $\Delta_{\text{reg}}$, then we will have
\begin{equation} \label{duistermaat-heckman-theorem-eq}
\omega_{\eta+\xi} = \omega_{\xi} - \langle \eta , c \rangle,
\end{equation}
where $c \in H^2(M_{\xi})$, or it is the Euler class of the $T^r$-fibration $\mu_T^{-1}(\xi) \rightarrow M_{\xi}$. 
\begin{lem} \label{symplectic-form-lemma}
\cite{audin1991,cho2014} If $\tilde{\omega}$ represents a cohomology class in $H^2_{T^r}(M,\mathbb{C})$, for any fixed component $F \in M^T$, we have $\tilde{\omega}|_{F}= \omega|_{F} + \langle \mu_T(F), \vec{\theta} \rangle$, where $\omega \in H^2(M,\mathbb{C})$ is the symplectic form of the original manifold $M$ and $\vec{\theta} = (\theta_1,\theta_2, \cdots, \theta_r)$ are the generators of $H^*(BT^r;\mathbb{C}) \cong \mathbb{C}[\theta_1,\theta_2, \cdots, \theta_r]$, or $\theta_k \in H^2(BS^1;\mathbb{C})$, for $k=1,\cdots,r$. In particular, if $F$ is an isolated point, then $\tilde{\omega}|_{F}= \langle \mu_T(F), \vec{\theta} \rangle$.
\end{lem}

Using the localization or the iterated residue algorithm described in the previous section \ref{localization-sub-sec}, we have
\[
\int_{M_{\xi}}{\kappa(\tilde{\omega})} = \sum_{P}{\int_{F_P}{C_P(\imath^*_{F_P}(\tilde{\omega}))}} = \sum_{P}{C_P(\tilde{\omega}|_{F_P})},
\]
since the fixed point set $M^T \equiv F$ contains isolated points $p_j$, for $j=1, \cdots, n+1$, in which the contribution $C_P$ for the path $P$ leading to $\mu_T(p_j)$ as a vertex of the moment polytope $\Delta$ is given by
\begin{equation} \label{path-cont-2-qubit-eq}
C_j(\tilde{\omega}|_F) = R_{0} \circ R_{1}(\tilde{\omega}|_{p_j}),
\end{equation}
where
\begin{equation} \label{path-residues-2-qubit-eq}
R_s(\tilde{\omega}|_{p_j}) = \text{res}_{\theta_{s+1}=0} \left( \frac{ \tilde{\omega}|_{p_j}}{e_T(\nu_s)} \right), \quad s=0,1,
\end{equation}
with $e_T(\nu_s)$, for $s=0,1$, given in Eq. \eqref{equivariant-euler-class-normal-s-eq}. Let's first suppose $\xi \in \Delta_{\text{reg}}$ belongs to the upper chamber (gray chamber in Figure. \ref{2-cube}). The dendrite starts from $\xi$ and intersects a co-dimension one wall at $q$ and in the last step reaches the fixed points $p_1$ and $p_3$. From the Eq. \eqref{a-matrix-r=2II}, the isotropy weights representation at $p_1$ is obtained as
\begin{equation} \label{2-qubit-p1-isotropy-weights-eq}
w_{2,1}^{(0)} = -2 \theta_2, \quad w_{3,1}^{(0)} = -2 \theta_1, \quad w_{4,1}^{(0)} = -2 \theta_1 - 2 \theta_2,
\end{equation}
whereas the non-zero isotropy representation at the intersection of the dendrite and the co-dimension one wall of $\Delta$, namely $q$, is given by
\begin{equation} \label{2-qubit-p1-isotropy-weights-eq2}
w_{2,1}^{(s=1)}= -2 \theta_2, \quad  w_{4,1}^{(s=1)} = -2 \theta_2.
\end{equation}
Now, let's consider a generic element $\gamma = (-2, -1)^T \in \mathfrak{t}_2$, with respect to whose $w_{l,j}^{(0)} ( \gamma ) \equiv \langle \gamma, w_{l,j}^{(0)} \rangle \neq 0$, for $l=1, \cdots \hat{j}, \cdots , r$ and since there are no negative isotropy weights $w_{l,j}^{(0)}$ with respect to the generic element $\gamma$ no polarization is required and from Eq. \eqref{2-qubit-p1-isotropy-weights-eq} and \eqref{2-qubit-p1-isotropy-weights-eq2} we have
\begin{equation} \label{2-qubit-euler-class-eq1}
e_T(\nu_0) = \prod_{l \neq 1}{\tilde{w}_{l,j}^{(0)}} = \prod_{l \neq 1}{w_{l,j}^{(0)}} = -8\theta_1^2 \theta_2 - 8 \theta_1 \theta_2^2,
\end{equation}
\begin{equation} \label{2-qubit-euler-class-eq2}
e_T(\nu_1) = \prod_{l \neq 1}^{2-1}{\tilde{w}_{l,j}^{(1)}} = \prod_{l \neq 1}{w_{l,j}^{(1)}} = 4 \theta_2^2.
\end{equation}

From the lemma \ref{symplectic-form-lemma}, we have
\begin{equation} \label{equivariant-symplectic-form-2-qubit-eq}
\tilde{\omega}|_{p_1} = \langle \mu_T(p_1) - \xi, \vec{\theta} \rangle = (1-\xi_1) \theta_1 + (1-\xi_2) \theta_2,
\end{equation}
where $\xi \equiv (\xi_1, \xi_2) \in \text{upper}(\Delta_{\text{reg}})$, and $\vec{\theta} = (\theta_1, \theta_2)^T \in \mathfrak{t}_2$. By replacing from the Eqs. \eqref{2-qubit-euler-class-eq1}, \eqref{2-qubit-euler-class-eq2} and \eqref{equivariant-symplectic-form-2-qubit-eq} in the iterated residue operations of Eq. \eqref{path-residues-2-qubit-eq} and then Eq. \eqref{path-cont-2-qubit-eq}, the contribution $C_1(\tilde{\omega}|_{p_1})$ of the dendrite's path starting from $\xi$ and ending at the $\mu_T(p_1)$ will be as
\begin{eqnarray} \label{path-contribution-2-qubit-p1-eq}
C_1(\tilde{\omega}|_{p_1}) & = & (\text{res})_{\theta_1=0} \left( \frac{1}{e_T(\nu_0)} (\text{res})_{\theta_2=0} \left( \frac{\tilde{\omega}|_{p_1}}{e_T(\nu_1)} \right) \right) \\ \nonumber
& = & (\text{res})_{\theta_1=0} \left( \frac{1}{-8\theta_1^2 \theta_2 - 8 \theta_1 \theta_2^2} (\text{res})_{\theta_2=0} \left( \frac{(1-\xi_1) \theta_1 + (1-\xi_2) \theta_2}{4 \theta_2^2} \right) \right) \\ \nonumber
&\propto & - (1 - \xi_2).
\end{eqnarray}
Repeating the same procedure, the contribution $C_3(\tilde{\omega}|_{p_3})$ of the dendrite's path starting from $\xi$ and ending at the $\mu_T(p_3)$, as in Figure. \ref{2-cube}, will be as
\begin{eqnarray} \label{path-contribution-2-qubit-p3-eq}
C_3(\tilde{\omega}|_{p_3}) & = & (\text{res})_{\theta_1=0} \left( \frac{1}{e_T(\nu_0)} (\text{res})_{\theta_2=0} \left( \frac{\tilde{\omega}|_{p_3}}{e_T(\nu_1)} \right) \right) \\ \nonumber
& = & (-1)(\text{res})_{\theta_1=0} \left( \frac{1}{-8\theta_1^2 \theta_2 + 8 \theta_1 \theta_2^2} (\text{res})_{\theta_2=0} \left( \frac{(-1-\xi_1) \theta_1 + (1-\xi_2) \theta_2}{4 \theta_2^2} \right) \right) \\ \nonumber
&\propto & - (1 - \xi_2).
\end{eqnarray}
where $\tilde{\omega}|_{p_3} = \langle \mu_T(p_3) - \xi, \vec{\theta} \rangle = (-1-\xi_1) \theta_1 + (1-\xi_2) 
\theta_2$ and, 
\[
e_T(\nu_0) = \prod_{l \neq 3}{\tilde{w}_{l,3}^{(0)}} = - 8\theta_1^2 \theta_2 + 8 \theta_1 \theta_2^2,
\] 
\[
e_T(\nu_1) = \prod_{l \neq 3}{\tilde{w}_{l,3}^{(1)}} = 4 \theta_2^2,
\]
since the set of isotropy weights representations at the point $\mu_T(p_3)$ contains two negative isotropy weights, namely $w_{1,3}^{(0)} = 2 \theta_1$ and $w_{2,3}^{(0)}  = 2 \theta_1 - 2 \theta_2$, whose polarization should be taken into account in finding the equivariant Euler classes $e_T(\nu_0)$ and $e_T(\nu_1)$ during the iterated residue algorithm. The multiplication by $(-1)$ after the second equality in Eq. \eqref{path-contribution-2-qubit-p3-eq} is due to the fact that in the second step of the path leading to $\mu_T(p_3)$ we have to move in the direction of $-\theta_1$. Therefore,
\begin{equation} \label{2-qubit-upper-chamber-integral-eq}
\int_{M_{\xi}}{\kappa(\tilde{\omega})} = C_1(\tilde{\omega}|_{p_1}) + C_3(\tilde{\omega}|_{p_3}) \propto - (1 - \xi_2),
\end{equation}
which holds for every $\xi \in \text{upper}(\Delta_{\text{reg}})$. Obviously, for a point $\xi \in \text{right}(\Delta_{\text{reg}})$, we will have  
\begin{equation} \label{2-qubit-right-chamber-integral-eq}
\int_{M_{\xi}}{\kappa(\tilde{\omega})} = C_1(\tilde{\omega}|_{p_1}) + C_2(\tilde{\omega}|_{p_2}),
\end{equation}
where $\mu_T(p_2)$ is the lower-right vertex of the moment polytope $\Delta$ at the coordinate $(\frac{1}{2},-\frac{1}{2})$. The contributions of the paths $C_1(\tilde{\omega}|_{p_1})$ and $C_2(\tilde{\omega}|_{p_2})$ can be computed similarly by following the iterated residue algorithms described previously. However, for points in the right (respectively left) chamber $\Delta_{\text{reg}}$, the order of the residue operations will be switched, since in the new dendrite the first step, which is intersecting a codimension one wall, would be in reverse order and the same argument hold for the other steps. More precisely,
\begin{eqnarray} \label{path-contribution-2-qubit-p1-eq2}
C_1(\tilde{\omega}|_{p_1}) & = & (\text{res})_{\theta_2=0} \left( \frac{1}{e_T(\nu_0)} (\text{res})_{\theta_1=0} \left( \frac{\tilde{\omega}|_{p_1}}{e_T(\nu_1)} \right) \right) \\ \nonumber
& = & (\text{res})_{\theta_2=0} \left( \frac{1}{-8\theta_1^2 \theta_2 - 8 \theta_1 \theta_2^2} (\text{res})_{\theta_1=0} \left( \frac{(1-\xi_1) \theta_1 + (1-\xi_2) \theta_2}{4 \theta_1^2} \right) \right) \\ \nonumber
&\propto& - (1 - \xi_1),
\end{eqnarray}
where 
\[
e_T(\nu_0) = \prod_{l \neq 1}{\tilde{w}_{l,1}^{(0)}} = -8\theta_1^2 \theta_2 - 8 \theta_1 \theta_2^2,
\] 
\[
e_T(\nu_1) = \prod_{l \neq 1}{\tilde{w}_{l,1}^{(1)}} = 4 \theta_1^2,
\]
since there are no negative isotropy weight with respect to the generic element $\gamma \in \mathfrak{t}_2$. Similarly, 
\begin{eqnarray} \label{path-contribution-2-qubit-p1-eq3}
C_2(\tilde{\omega}|_{p_2}) & = & (\text{res})_{\theta_2=0} \left( \frac{1}{e_T(\nu_0)} (\text{res})_{\theta_1=0} \left( \frac{\tilde{\omega}|_{p_2}}{e_T(\nu_1)} \right) \right) \\ \nonumber
& = & (\text{res})_{\theta_2=0} \left( \frac{1}{-8\theta_1^2 \theta_2 + 8 \theta_1 \theta_2^2} (\text{res})_{\theta_1=0} \left( \frac{(1-\xi_1) \theta_1 + (1-\xi_2) \theta_2}{4 \theta_1^2} \right) \right) \\ \nonumber
&\propto& - (1 - \xi_1),
\end{eqnarray}
where 
\[
e_T(\nu_0) = \prod_{l \neq 2}{\tilde{w}_{l,2}^{(0)}} = -8\theta_1^2 \theta_2 + 8 \theta_1 \theta_2^2,
\] 
\[
e_T(\nu_1) = \prod_{l \neq 1}{\tilde{w}_{l,2}^{(1)}} = 4 \theta_1^2,
\]
since there is only one negative isotropy weight $w_{l,2}^{(0)}$. Therefore, 
\begin{equation} \label{2-qubit-right-chamber-integral-eq2}
\int_{M_{\xi}}{\kappa(\tilde{\omega})} = C_1(\tilde{\omega}|_{p_1}) + C_2(\tilde{\omega}|_{p_2}) \propto - (1 - \xi_1).
\end{equation}

We can repeat the same procedure for $\xi \in \text{below}(\Delta_{\text{reg}})$ and $\text{left}(\Delta_{\text{reg}})$ and obtain similar results. In summary, for the case $r=2$ we have
\begin{equation} \label{2-qubits-final-eq}
\int_{M_{\xi}}{\kappa(\tilde{\omega})} = c \left( 1 - \text{max}(|\xi_1|,|\xi_2|) \right),
\end{equation}
where $c$ is a constant. In fact, the result in Eq. \eqref{2-qubits-final-eq} satisfies concavity property and coincides with the Abelian Duistermaat-Heckman measure obtained in \cite{christandl2014}, since $\text{dim}(M_{\xi})=2$ and $\kappa(\tilde{\omega}) \in H^2(M_{\xi})$.

\subsubsection{Three Qubits} \label{three-qubits-example-integral}
Now, let's consider again the case $r=3$, in which the torus $T^2=S^1 \times S^1 \times S^1$ acts in a Hamiltonian fashion on the K\"ahler manifold $\mathbb{P}_7$ through the homomorphism \eqref{action}. The matrix $a_{i,j}\equiv A$ is given in the Eq. \eqref{a-matrix-r=3} and the moment polytope $\Delta$ is a cube with $8$ vertices as the image under the moment map $\mu_T : \mathbb{P}_7 \rightarrow \mathfrak{t}^*$ of the fixed points set $M^T$ and the symplectic quotient $M_{\xi} = \mu_T^{-1}(\xi)/T^3$, with $\xi \in \text{upper}(\Delta_{\text{reg}})$ as shown in Figure. \ref{3-cube-example}, possesses at most orbifold singularities.

\begin{figure}[h]
\centering
\begin{tikzpicture}[
  vector/.style={ultra thick,black,>=stealth,->},
  atom/.style={blue}, x=4cm,y=3cm,z=1.5cm
  ]
    \coordinate (p0) at (0,0,0);
    \coordinate (p1) at (1,0,0);
    \coordinate (p2) at (0,1,0);
    \node[left] at (p2) {$\mu_T(p_3)=(\frac{1}{2},-\frac{1}{2},\frac{1}{2})$};
    \coordinate (p3) at (0,0,1);

    \coordinate (p4) at (1,1,0);
    \node[right=.5em] at (p4) {$\mu_T(p_1) =(\frac{1}{2},\frac{1}{2},\frac{1}{2})$};
    \fill[black] (p4) circle (2pt);
    \coordinate (p5) at (1,0,1);
    \coordinate (p6) at (0,1,1);
    \coordinate (p7) at (1,1,1);
    \node[right] at (p7) {$\mu_T(p_4)=(-\frac{1}{2},\frac{1}{2},\frac{1}{2})$};
    \fill[black] (p7) circle (2pt);

    \coordinate (a2) at (1,1,0.5);
    \node[right=.35em] at (a2) {$q_2$};
    \fill[black] (a2) circle (2pt);    
    \draw[ultra thick,->] (a2) -- (p7);
    \draw[ultra thick,->] (a2) -- (p4);
    
    \coordinate (a3) at (0.65,1,0.5);
    \node[above] at (a3) {$q_1$};
    \fill[black] (a3) circle (2pt);
    \draw[ultra thick,->] (a3) -- (a2);

    \coordinate (a4) at (0.65,0.7,0.5);
    \node[below] at (a4) {$\xi$};
    \fill[black] (a4) circle (2pt);
    \draw[ultra thick,->] (a4) -- (a3);

    \coordinate (a5) at (0.5,0.5,0.5);
    \coordinate (a6) at (1.5,0.5,0.5);
    \node[right] at (a6) {$\theta_2$};
    \coordinate (a7) at (0.5,1.5,0.5);
    \node[above] at (a7) {$\theta_3$};
    \coordinate (a8) at (0.5,0.5,-1.25);
    \node[left] at (a8) {$\theta_1$};
    \draw[dotted,->] (a5) -- (a6);
    \draw[dotted,->] (a5) -- (a7);
    \draw[dotted,->] (a5) -- (a8);
    
    \coordinate (a9) at (0.5,1,0);
    \node[above=.35em] at (a9) {$q^{\prime}_2$};
    \fill[black] (a9) circle (2pt);    
    \draw[ultra thick,->] (a9) -- (p2);
    \draw[ultra thick,->] (a9) -- (p4);
    
    \coordinate (a10) at (0.5,0.75,0);
    \node[below] at (a10) {$q^{\prime}_1$};
    \fill[black] (a10) circle (2pt);
    \draw[ultra thick,->] (a10) -- (a9);
    
    \coordinate (a11) at (0.5,0.75,0.25);
    \node[right=.15em] at (a11) {$\xi^{\prime}$};
    \fill[black] (a11) circle (2pt);
    \draw[ultra thick,->] (a11) -- (a10);

\fill[gray,opacity=0.2] (p7) -- (p6) -- (p2) -- (p4);

    \draw[dashed] (p0) -- (p1);
    \draw[dashed] (p0) -- (p2);
    \draw[dashed] (p0) -- (p3);
    \draw[dashed] (p0) -- (p3);
    \draw[dashed] (p3) -- (p5);
    \draw[dashed] (p3) -- (p6);
    \draw[dashed] (p6) -- (p2);
    \draw[dashed] (p5) -- (p1);
    \draw[dashed] (p5) -- (p7);
    \draw[dashed] (p1) -- (p4);
    \draw[dashed] (p6) -- (p7);
    \draw[dashed] (p2) -- (p4);

\end{tikzpicture}
\caption{$\mu_T(\mathbb{P}_7)\equiv \Delta_3$, for $M = \mathbb{P}_7$ with $T^3$ Hamiltonian action. Here, $\xi$ belongs to the upper $(\Delta_{\text{reg}})$ and $\xi^{\prime}$ to the front$(\Delta_{\text{reg}})$. The black arrows illustrate the corresponding dendrites.}
\label{3-cube-example} 
\end{figure}
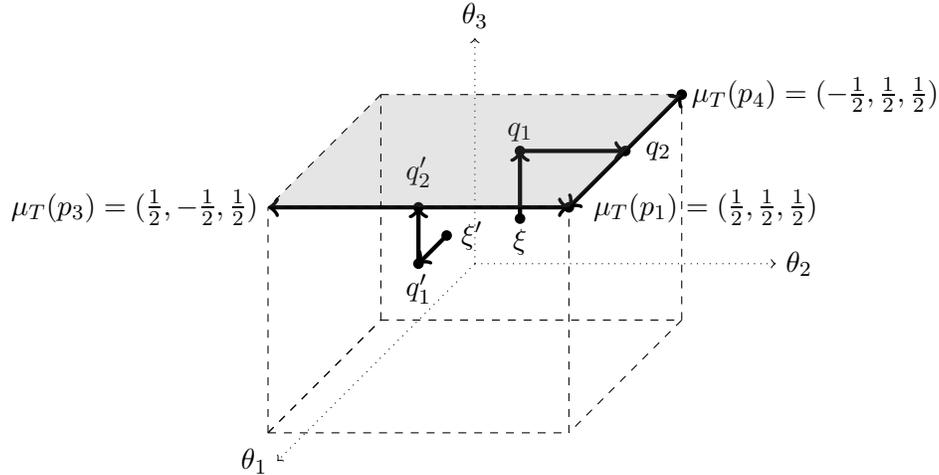
Recall that $2m= \text{dim}_{\mathbb{R}}(M_{\xi}) = 8$ and the goal is to find the pairing $\int_{M_{\xi}}{\kappa(\tilde{\eta}^a \tilde{\omega}^b)}$, where $a$ and $b$ are non-negative integers such that $a+b =m$ and $\tilde{\omega} := \omega + d \langle \mu(p), X \rangle \in H^2_{T^3}(\mathbb{P}_7;\mathbb{C})$, for every $X \in \mathfrak{t}$. The polarization direction $\gamma \in \mathfrak{t}$ is chosen to be $(-4,-2,-1)^T$. For $a=0$ the integral $\int_{M_{\xi}}{\kappa(\tilde{\omega}^m)}$, where $\tilde{\omega}^m = \tilde{\omega} \wedge \cdots \wedge \tilde{\omega}$ and $\kappa(\tilde{\omega}^m) \cong \omega_{\xi}^m \in H^{2m}(M_{\xi})$, can be computed using the iterated residue algorithm described in the section \ref{iterated-residue-alg} as follows:
\[
\int_{M_{\xi}}{\kappa(\tilde{\omega}^m)} = \sum_{P}{\int_{F_P}{C_P(\imath^*_{F_P}(\tilde{\omega}^m))}} = \sum_{P}{C_P(\tilde{\omega}^m|_{F_P})},
\]
where 
\begin{equation} \label{path-cont-3-qubit-eq}
C_j(\tilde{\omega}^m|_F) = R_{0} \circ R_{1} \circ R_{2}(\tilde{\omega}^m|_{p_j}),
\end{equation}
\begin{equation} \label{path-residues-3-qubit-eq}
R_s(\tilde{\omega}^m|_{p_j}) = \text{res}_{\theta_{s+1}=0} \left( \frac{ \tilde{\omega}^m|_{p_j}}{e_T(\nu_s)} \right), \quad s=0,1,2,
\end{equation}
with $e_T(\nu_s)$ given in Eq. \eqref{equivariant-euler-class-normal-s-eq}. Using the lemma \ref{symplectic-form-lemma} we have $\tilde{\omega}^m|_F = \langle \mu_T(F) - \xi, \theta \rangle^m \in H^{2m}_T(M;\mathbb{C})$, since the fixed point set $M^T$ contains $n+1$ isolated fixed points. In our case and for $r=3$ we have $m=4$ and for $\xi \equiv (\xi_1,\xi_2,\xi_3) \in \text{upper}(\Delta_{\text{reg}})$ in Figure \ref{3-cube-example} the dendrite starts from $\xi$ and in the first step intersects co-dimension one wall at $q_1$ and in the second step intersects the co-dimension two wall at $q_2$ before reaching the fixed points $\mu_T(p_1)$ and $\mu_T(p_4)$ in the last step. Therefore,
\begin{equation} \label{3-qubits-upper-reg-int}
\int_{M_{\xi}}{\kappa(\tilde{\omega}^4)} = C_1(\tilde{\omega}^4|_{p_1}) + C_4(\tilde{\omega}^4|_{p_4}),
\end{equation}
where
\[
C_1(\tilde{\omega}^4|_{p_1}) = (\text{res})_{\theta_1=0} \left( \frac{1}{e_T(\nu_0)} (\text{res})_{\theta_2=0} \left( \frac{1}{e_T(\nu_1)} (\text{res})_{\theta_3=0} \left( \frac{\tilde{\omega}^4|_{p_1}}{e_T(\nu_2)} \right) \right) \right),
\]
\[
C_4(\tilde{\omega}^4|_{p_4}) = (-1)(\text{res})_{\theta_1=0} \left( \frac{1}{e_T(\nu_0)} (\text{res})_{\theta_2=0} \left( \frac{1}{e_T(\nu_1)} (\text{res})_{\theta_3=0} \left( \frac{\tilde{\omega}^4|_{p_4}}{e_T(\nu_2)} \right) \right) \right),
\]
in which the multiplication by $(-1)$ in the second equation is due to the fact that in the last step of the path leading to $\mu_T(p_4)$ we have to move in the direction of $-\theta_1$ and
\begin{eqnarray} \label{3-qubits-restrictions-1}
\tilde{\omega}^4|_{p_1} & = & \langle \mu_T(p_1) - \xi, \vec{\theta} \rangle^4 = \left( (1-\xi_1)\theta_1 + (1-\xi_2) \theta_2 + (1-\xi_3) \theta_3 \right)^4, \\ \nonumber
\tilde{\omega}^4|_{p_4} & = &\langle \mu_T(p_4) - \xi, \vec{\theta} \rangle^4 = \left( (-1-\xi_1)\theta_1 + (1-\xi_2) \theta_2 + (1-\xi_3) \theta_3 \right)^4.
\end{eqnarray}
To find $C_1(\tilde{\omega}^4|_{p_1})$ we have to note that there are no negative isotropy weights and so no polarization of isotropy weights is required. Therefore,
\[
e_T(\nu_s) = \prod_{l \neq 1}^{2^3 - 2^s}{w_{l,1}^{(s)}}, \quad s=0,1,2,
\]
and hence the contributions from the path $C_1$ leading to $\mu_T(p_1)$ can be calculated as
\begin{equation} \label{3-qbits-upper-reg-c1-cont}
C_1(\tilde{\omega}^4|_{p_1}) \propto (1-\xi_2)(-1+\xi_3)^3.
\end{equation}
For $C_4(\tilde{\omega}^4|_{p_4})$, we have to note that there are four negative isotropy weights, namely $w_{l,4}^{(0)}$, for $l=1,2,3,5$, whose polarization have to be taken into account to calculate $C_4(\tilde{\omega}^4|_{p_4})$ as follows
\begin{equation} \label{3-qbits-upper-reg-c4-cont}
C_4(\tilde{\omega}^4|_{p_4}) \propto (1-\xi_2)(-1+\xi_3)^3.
\end{equation}
Hence, from Eq. \eqref{3-qubits-upper-reg-int}, we will have 
\begin{equation} \label{3-qbits-upper-reg-c1+c4-cont}
\int_{M_{\xi}}{\kappa(\tilde{\omega}^4)} = C_1(\tilde{\omega}^4|_{p_1}) + C_4(\tilde{\omega}^4|_{p_4}) = c (1-\xi_2)(-1+\xi_3)^3,
\end{equation}
which is in fact the Abelian Duistermaat-Heckman measure for a three-qubit system in its pure state for every $\xi \equiv (\xi_1,\xi_2,\xi_3) \in \text{upper}(\Delta_{\text{reg}})$ in Figure \ref{3-cube-example}. We can repeat the above calculations for moment values $\xi \in \Delta_{\text{reg}}$ belonging to other regular chambers of the hypercube and obtain similar results. For instance, let $\xi^{\prime} \equiv (\xi^{\prime}_1,\xi^{\prime}_2,\xi^{\prime}_3) \in \text{front}(\Delta_{\text{reg}})$, then
\begin{equation} \label{3-qubits-front-reg-int}
\int_{M_{\xi^{\prime}}}{\kappa(\tilde{\omega}^4)} = C_1(\tilde{\omega}^4|_{p_1}) + C_3(\tilde{\omega}^4|_{p_3}),
\end{equation}
where
\[
C_1(\tilde{\omega}^4|_{p_1}) = (\text{res})_{\theta_2=0} \left( \frac{1}{e_T(\nu_0)} (\text{res})_{\theta_3=0} \left( \frac{1}{e_T(\nu_1)} (\text{res})_{\theta_1=0} \left( \frac{\tilde{\omega}^4|_{p_1}}{e_T(\nu_2)} \right) \right) \right),
\]
\[
C_3(\tilde{\omega}^4|_{p_3}) = (-1)(\text{res})_{\theta_2=0} \left( \frac{1}{e_T(\nu_0)} (\text{res})_{\theta_3=0} \left( \frac{1}{e_T(\nu_1)} (\text{res})_{\theta_1=0} \left( \frac{\tilde{\omega}^4|_{p_3}}{e_T(\nu_2)} \right) \right) \right),
\]
in which the multiplication by $(-1)$ in the second equation is due to the fact that in the last step of the path leading to $\mu_T(p_3)$ we have to move in the direction of $-\theta_2$ and similar to Eq. \eqref{3-qubits-restrictions-1}
\begin{eqnarray} \label{3-qubits-restrictions-2}
\tilde{\omega}^4|_{p_1} & = & \langle \mu_T(p_1) - \xi^{\prime}, \vec{\theta} \rangle^4 = \left( (1-\xi^{\prime}_1)\theta_1 + (1-\xi^{\prime}_2) \theta_2 + (1-\xi^{\prime}_3) \theta_3 \right)^4, \\ \nonumber
\tilde{\omega}^4|_{p_3} & = &\langle \mu_T(p_3) - \xi^{\prime}, \vec{\theta}  \rangle^4 = \left( (1-\xi^{\prime}_1)\theta_1 + (-1-\xi^{\prime}_2) \theta_2 + (1-\xi^{\prime}_3) \theta_3 \right)^4.
\end{eqnarray}
Recall that
\[
e_T(\nu_s) = \prod_{l \neq 1}^{2^3 - 2^s}{w_{l,j}^{(s)}}, \quad s=0,1,2,
\]
where $j=1$ and $j=3$ in $C_1$ and $C_3$, respectively. Since there is no negative isotropy weights at $\mu_T(p_1)$, the contribution from the path $C_1$ can be calculated as
\begin{equation} \label{3-qbits-front-reg-c1-cont}
C_1(\tilde{\omega}^4|_{p_1}) \propto (1-\xi_3)(-1+\xi_1)^3.
\end{equation}
At $\mu_T(p_3)$, there are four negative isotropy weights, namely $w_{l,3}^{(0)}$ with $l=1,2,5,6$ whose polarization again have to be taken into account to find $C_3(\tilde{\omega}^4|_{p_3})$ as
\begin{equation} \label{3-qbits-front-reg-c3-cont}
C_3(\tilde{\omega}^4|_{p_3}) \propto (1-\xi_3)(-1+\xi_1)^3.
\end{equation} 
Hence, from Eq. \eqref{3-qubits-front-reg-int}, we will have 
\begin{equation} \label{3-qbits-front-reg-c1+c3-cont}
\int_{M_{\xi^{\prime}}}{\kappa(\tilde{\omega}^4)} = C_1(\tilde{\omega}^4|_{p_1}) + C_3(\tilde{\omega}^4|_{p_3}) \propto (1-\xi_3)(-1+\xi_1)^3,
\end{equation}
which is in fact the Abelian Duistermaat-Heckman measure for a three-qubit system in its pure state for every $\xi^{\prime} \equiv (\xi^{\prime}_1,\xi^{\prime}_2,\xi^{\prime}_3) \in \text{front}(\Delta_{\text{reg}})$ in Figure \ref{3-cube-example}

Now, let $a \neq 0$ and find the pairing $\int_{M_{\xi}}{\kappa(\tilde{\eta}^a \tilde{\omega}^b)}$, where $a$ and $b$ are non-negative integers of complementary degrees $a+b =m$. Again, by using the lemma \ref{symplectic-form-lemma} we have $\tilde{\omega}^b|_F = \langle \mu_T(F) - \xi, \vec{\theta} \rangle^b \in H^{2b}_T(M;\mathbb{C})$ and $\tilde{\eta}^a|_F = (-\sum_{i}{\theta_i})^a \in H^{2a}_T(M;\mathbb{C})$, where $\theta_i$s are the generators of $H^*(BS^1;\mathbb{C})$ such that the Euler class of the Hopf bundle $ET^r \rightarrow BT^r$ is $-\vec{\theta} \equiv - (\theta_1, \theta_2, \cdots, \theta_r)$ and since the fixed point set $M^T$ contains $n+1$ isolated fixed points \cite{cho2014}. Let's consider the three-qubit example and choose the same polarization direction $\gamma = (-4,-2,-1)^T \in \mathfrak{t}$ and follow the iterated residue algorithm in section \ref{iterated-residue-alg} to compute
\[
\int_{M_{\xi}}{\kappa(\tilde{\eta}^a \tilde{\omega}^b)} = \sum_{P}{\int_{F}{C_P(\imath^*_{F}(\tilde{\eta}^a \tilde{\omega}^b))}} = \sum_{P}{C_P(\tilde{\eta}^a|_F \tilde{\omega}^b|_F)},
\]
where 
\begin{equation} \label{path-cont-3-qubit-eq2}
C_j(\tilde{\eta}^a|_F \tilde{\omega}^b|_F) = R_{0} \circ R_{1} \circ R_{2}(\tilde{\eta}^a|_{p_j} \tilde{\omega}^b|_{pj}),
\end{equation}
\begin{equation} \label{path-residues-3-qubit-eq2}
R_s(\tilde{\eta}^a|_{p_j} \tilde{\omega}^b|_{p_j}) = \text{res}_{\theta_{s+1}=0} \left( \frac{ \tilde{\eta}^a|_{p_j} \tilde{\omega}^b|_{p_j}}{e_T(\nu_s)} \right), \quad s=0,1,2,
\end{equation}
with $e_T(\nu_s)$ given in Eq. \eqref{equivariant-euler-class-normal-s-eq}. For three-qubit example $a+b=m=4$ and for instance for $\xi \equiv (\xi_1,\xi_2,\xi_3) \in \text{upper}(\Delta_{\text{reg}})$, whose corresponding dendrite is shown in Figure \ref{3-cube-example}, we will have
\begin{equation} \label{3-qubits-upper-reg-int2}
\int_{M_{\xi}}{\kappa(\tilde{\eta}^a \tilde{\omega}^b)} = C_1(\tilde{\eta}^a|_{p_1} \tilde{\omega}^b|_{p_1}) + C_4(\tilde{\eta}^a|_{p_4} \tilde{\omega}^b|_{p_4}),
\end{equation}
where
\[
C_1(\tilde{\eta}^a|_{p_1} \tilde{\omega}^b|_{p_1}) = (\text{res})_{\theta_1=0} \left( \frac{1}{e_T(\nu_0)} (\text{res})_{\theta_2=0} \left( \frac{1}{e_T(\nu_1)} (\text{res})_{\theta_3=0} \left( \frac{\tilde{\eta}^a|_{p_1} \tilde{\omega}^b|_{p_1}}{e_T(\nu_2)} \right) \right) \right),
\]
\[
C_4(\tilde{\eta}^a|_{p_4} \tilde{\omega}^b|_{p_4}) = (-1)(\text{res})_{\theta_1=0} \left( \frac{1}{e_T(\nu_0)} (\text{res})_{\theta_2=0} \left( \frac{1}{e_T(\nu_1)} (\text{res})_{\theta_3=0} \left( \frac{\tilde{\eta}^a|_{p_4} \tilde{\omega}^b|_{p_4}}{e_T(\nu_2)} \right) \right) \right),
\]
with
\[
\tilde{\eta}^a|_{p_1} \tilde{\omega}^b|_{p_1} =  (-\sum_{i=1}^{3}{\theta_i})^a \langle \mu_T(p_1) - \xi, \vec{\theta} \rangle^b = (-\sum_{i=1}^{3}{\theta_i})^a \left( (1-\xi_1)\theta_1 + (1-\xi_2) \theta_2 + (1-\xi_3) \theta_3 \right)^b, 
\]
\[
\tilde{\eta}^a|_{p_4} \tilde{\omega}^b|_{p_4} =  (-\sum_{i=1}^{3}{\theta_i})^a \langle \mu_T(p_4) - \xi, \vec{\theta} \rangle^b = (-\sum_{i=1}^{3}{\theta_i})^a \left( (-1-\xi_1)\theta_1 + (1-\xi_2) \theta_2 + (1-\xi_3) \theta_3 \right)^b.
\]
Hence, from Eq. \eqref{3-qubits-upper-reg-int2} and for instance for $a=1$ and $b=3$ we have
\begin{equation} \label{3-qbits-upper-reg-c1+c4-cont-a-b}
\int_{M_{\xi}}{\kappa(\tilde{\eta} \, \tilde{\omega}^3)} = C_1(\tilde{\eta}|_{p_1} \tilde{\omega}^3|_{p_1}) + C_4(\tilde{\eta}|_{p_4} \tilde{\omega}^3|_{p_4}) = c (\xi_1 + \xi_2 + \xi_3)(-1+\xi_3)^3.
\end{equation}
Similar results can be obtained for different values of $a$ and $b$, provided that $a+b =m = \text{dim}_{\mathbb{C}}(M_{\xi})$.

\section{Conclusions} \label{conclusions}

In the first part of the paper, we proposed an algorithm on how to explicitly obtain cohomology rings of the Abelian symplectic reduced spaces for Hamiltonian action of maximal torus of the Local Unitary group on the complex projective manifold of pure multi-qubit quantum systems in terms of elementary symmetric functions. Then, by studying the geometry of associated moment polytope, we employed a recursive wall-crossing algorithm to compute cohomological pairings on the corresponding Abelian symplectic quotients for regular values of the torus moment map. The computations were then elaborated by detailed examples for two and three-qubit cases.

From a physical perspective, the Abelianization corresponds to the problem of joint probability distributions of classical marginals with only diagonal entries of local density matrices in a multi-particle quantum system. In \cite{christandl2014}, the authors used two algorithms, namely the Heckman and the single-summand algorithms, to compute both Abelian and non-Abelian Duistermaat-Heckman measures for the same type of Hamiltonian group actions as this paper. However, in this paper we emphasize on the role played by equivariant cohomology to find both the cohomology rings and cohomological pairings over the Abelian symplectic reduced spaces.

Representation-theoretically, quantized version of the results in this paper are related to the multiplicity functions and branching coefficients for compact Lie group representations (for more details, readers can refer to \cite{vergne2015}). However, generalization of the results in this paper to the non-Abelian case require the knowledge of intersection cohomology \cite{goresky1980,goresky1983}, since the resulting non-Abelian symplectic reduced spaces would be stratified symplectic quotients \cite{sjamaar1991}. Such kind of generalizations are subjects of future studies. 

\section{Acknowledgments} \label{acknowledgments}
The author would like to thank Matthias Christandl, Beno\^{i}t Collins and Michael Walter for fruitful discussions.

\printbibliography

\vspace{.2in}
\noindent
Department of Mathematics and Statistics, University of Ottawa, 585 King Edward, Ottawa ON, K1N 6N5, Canada,
\textit{E-mail:} \href{mailto:smollada@uottawa.ca}{\texttt{smollada@uottawa.ca}}

\end{document}